\documentclass[journal,twocolumn,twoside]{IEEEtran}

\usepackage{amsfonts}
\usepackage{color}
\usepackage{graphicx}
\usepackage[dvips]{epsfig}
\usepackage{graphics} 
\usepackage{times} 
\usepackage[cmex10]{amsmath} 
\usepackage{amssymb}  
\usepackage{cite}
\usepackage{multirow}
\usepackage[tight,footnotesize]{subfigure}
\usepackage{algorithm}
\usepackage{amsthm}
\usepackage{mathtools}
\usepackage{url}

\def\norm #1{\left\|#1\right\|}

\def\twon #1{\left\|#1\right\|_2}
\def\onen #1{\left\|#1\right\|_1}

\def\frobn #1{\left\|#1\right\|_{\text{F}}}

\def\abs #1{\left|#1\right|}

\def\st{\text{subject to }}

\def\bC{\mathbb{C}}

\def\bR{\mathbb{R}}
\def\bE{\mathbb{E}}

\def\bT{\mathbb{T}}

\def\bZ{\mathbb{Z}}

\def\m #1{\boldsymbol{#1}}

\def\cA{\mathcal{A}}
\def\cC{\mathcal{C}}

\def\cM{\mathcal{M}}

\def\cP{\mathcal{P}}

\def\bee{\begin{equation}}
\def\ene{\end{equation}}

\def\beq{\begin{eqnarray}}
\def\enq{\end{eqnarray}}
\def\lentwo{\setlength\arraycolsep{2pt}}

\newtheorem{lem}{Lemma}
\newtheorem{rem}{Remark}

\newtheorem{thm}{Theorem}
\newtheorem{prop}{Proposition}

\newcommand{\BOX}{\hfill\rule{2mm}{2mm}}

\def\equ #1{\begin{equation}#1\end{equation}}
\def\equa #1{\begin{eqnarray}#1\end{eqnarray}}
\def\sbra #1{\left(#1\right)}
\def\mbra #1{\left[#1\right]}
\def\lbra #1{\left\{#1\right\}}
\def\diag #1{\text{diag}#1}
\def\tr #1{\text{tr}#1}

\def\rank #1{\text{rank}#1}
\def\st {\text{ subject to }}

\title{Vandermonde Decomposition of Multilevel Toeplitz Matrices With Application to Multidimensional Super-Resolution}
\author{Zai Yang, \emph{Member, IEEE,} Lihua Xie, \emph{Fellow, IEEE,} and Petre Stoica, \emph{Fellow, IEEE}
\thanks{This work appeared in part in the {\em Proceedings of the 2015 IEEE International Symposium on Information Theory (ISIT)}, Hong Kong, China, June 2015 \cite{yang2015generalized}.

Z. Yang is with the School of Automation, Nanjing University of Science and Technology, Nanjing 210094, China, and with the School of Electrical and Electronic Engineering, Nanyang Technological University, Singapore 639798 (e-mail: yangzai@ntu.edu.sg).

L. Xie is with the School of Electrical and Electronic Engineering, Nanyang Technological University, Singapore 639798 (e-mail: elhxie@ntu.edu.sg).

P. Stoica is with the Department of Information Technology, Uppsala University, Uppsala, SE 75105, Sweden (e-mail: ps@it.uu.se).}}

\begin{document}
\maketitle

\begin{abstract}
The Vandermonde decomposition of Toeplitz matrices, discovered by Carath\'{e}odory and Fej\'{e}r in the 1910s and rediscovered by Pisarenko in the 1970s, forms the basis of modern subspace methods for 1D frequency estimation. Many related numerical tools have also been developed for multidimensional (MD), especially 2D, frequency estimation; however, a fundamental question has  remained unresolved as to whether an analog of the Vandermonde decomposition holds for multilevel Toeplitz matrices in the MD case. In this paper, an affirmative answer to this question and a constructive method for finding the decomposition are provided when the matrix rank is lower than the dimension of each Toeplitz block. A numerical method for searching for a decomposition is also proposed when the matrix rank is higher. The new results are applied to studying MD frequency estimation within the recent super-resolution framework. A precise formulation of the atomic $\ell_0$ norm is derived using the Vandermonde decomposition. Practical algorithms for frequency estimation are proposed based on relaxation techniques. Extensive numerical simulations are provided to demonstrate the effectiveness of these algorithms compared to the existing atomic norm and subspace methods.
\end{abstract}


\section{Introduction}

The Vandermonde decomposition is a classical result by Carath\'{e}odory and Fej\'{e}r dating back to 1911 \cite{caratheodory1911zusammenhang}. To be specific, suppose that $\m{T}$ is an $n\times n$ positive semidefinite (PSD) Toeplitz matrix of rank $r<n$. The result states that $\m{T}$ can be uniquely decomposed as
\equ{\m{T} = \m{A}\m{P}\m{A}^H, \label{eq:vanderdecp1}}
where $\m{P}$ is an $r\times r$ positive definite diagonal matrix and $\m{A}$ is an $n\times r$ Vandermonde matrix whose columns correspond to uniformly sampled complex sinusoids with different frequencies. The result became important in the area of data analysis and signal processing when it was rediscovered by Pisarenko and used for frequency retrieval from the data covariance matrix \cite{pisarenko1973retrieval}. From then on, the Vandermonde decomposition, also referred to as the Carath\'{e}odory-Fej\'{e}r-Pisarenko decomposition, has formed the basis of a prominent subset of methods designated as subspace methods, e.g., multiple signal classification (MUSIC) and estimation of parameters by rotational invariant techniques (ESPRIT) (see the review in \cite{stoica2005spectral}).

The problem of estimating multidimensional (MD) frequencies arises in various applications including array processing, radar, sonar, astronomy and medical imaging. Inspired by the results in the 1D case, several computational subspace methods have been proposed for 2D frequency estimation such as 2D MUSIC \cite{hua1993pencil}, 2D ESPRIT \cite{haardt19952d,li1992two}, matrix enhancement and matrix pencil (MEMP) \cite{hua1992estimating} and the multidimensional folding (MDF) techniques \cite{liu2002almost,liu2006eigenvector, liu2007multidimensional}. However, a fundamental question remains unresolved as to whether an analog of the Vandermonde decomposition result holds true in the 2D or more general in the MD case. Note that the data covariance matrix corresponding to MD frequency estimation is a multilevel Toeplitz (MLT) matrix (see the definition in the next section). Consequently, the question can be phrased as follows:
\begin{quote}Given a PSD, rank-deficient MLT matrix, does it always admit a Vandermonde-like decomposition parameterized by MD frequencies? In other words, can this matrix be always the covariance matrix of an MD sinusoidal signal?
\end{quote}

An answer to the above question has recently become important due to the super-resolution framework in \cite{candes2013towards} which studies the recovery of fine details in a sparse (1D or MD) frequency spectrum from coarse scale time-domain samples. With compressive measurements super-resolution actually generalizes the compressed sensing problem in \cite{candes2006robust} to the continuous (as opposed to discretized/gridded) frequency setting and is referred to as off-grid or continuous compressed sensing \cite{tang2012compressed,yang2014continuous}. The paper \cite{candes2013towards} proposed a convex optimization method based on the atomic norm (or total variation norm, see \cite{aleksanyan1944real,chandrasekaran2012convex}) and proved that in the noiseless case the frequencies can be recovered with infinite precision provided that they are sufficiently separated. Unlike other methods, the atomic norm method is stable in the presence of noise and can deal with missing data \cite{candes2013super,tang2012compressed,bhaskar2013atomic,yang2015gridless,tang2015near}. It was also extended to the multi-snapshot case in array processing \cite{yang2014continuous,yang2014exact}. Moreover, a reweighted atomic norm method with enhanced sparsity and resolution was proposed in \cite{yang2014enhancing}.

The atomic norm is a continuous counterpart of the $\ell_1$ norm. A finite-dimensional formulation of it is required for numerical computations. In the 1D case a semidefinite program (SDP) formulation was provided based on the Vandermonde decomposition of Toeplitz matrices \cite{tang2012compressed,yang2014exact}. However, in the MD case a similar result has not been available as an analog of the Vandermonde decomposition was unknown. Interestingly, an SDP formulation with unspecified parameters was derived in \cite{xu2014precise} based on duality and the bounded real lemma for multivariate trigonometric polynomials \cite{dumitrescu2007positive}, which indeed is related to an MLT matrix whose dimension though is left unspecified.\footnote{Strictly speaking, the SDP formulation of the atomic norm given in \cite{xu2014precise} is a relaxed version since it is based on the so-called sum-of-squares relaxation for nonnegative multivariate trigonometric polynomials \cite{dumitrescu2007positive}.} A relaxed version of this formulation has also been applied in \cite{bendory2015super,heckel2014super} to the 2D case.

In this paper, we generalize the Carath\'{e}odory-Fej\'{e}r's result from the 1D to the MD case and provide an affirmative answer to the question asked above in the case when the matrix rank is lower than the dimension of each Toeplitz block. The new matrix $\m{A}$ in the resulting decomposition (see \eqref{eq:vanderdecp1}), which is still called Vandermonde, is the Khatri-Rao product of several Vandermonde matrices. A constructive method is provided for finding the Vandermonde decomposition. When the matrix rank is higher a numerical approach is also proposed that is guaranteed to find a Vandermonde decomposition if some conditions are satisfied.

To demonstrate the usefulness of the Vandermonde decomposition presented in this paper, we study the MD super-resolution problem with compressive measurements. A precise formulation of the atomic $\ell_0$ norm is derived based on the decomposition. Practical algorithms for solving the atomic $\ell_0$ norm minimization problem are proposed based on convex relaxation as well as on nonconvex relaxation and reweighted minimization. Frequency retrieval is finally accomplished using the proposed Vandermonde decomposition algorithms. Numerical results are provided to demonstrate the advantage of the proposed solutions over the state-of-the-art.

\subsection{Connections to Prior Art}

Similarly to this paper that generalizes the Carath\'{e}odory-Fej\'{e}r's Vandermonde decomposition from the 1D to the MD case, other related generalizations have also been attempted in the literature. In \cite{sidiropoulos2001generalizing} the uniqueness part of the Carath\'{e}odory-Fej\'{e}r's result was generalized to the MD case. Different from our result, \cite{sidiropoulos2001generalizing} focused on the identifiability of parameters of MD complex exponentials given discrete samples of their superposition. An analogous decomposition was presented in \cite{georgiou2000signal} for state-covariance matrices (including the Toeplitz case) of stable linear filters driven by certain time-series. The author also studied its multivariable counterpart in \cite{georgiou2007caratheodory}. Though the state-covariance matrices in \cite{georgiou2007caratheodory} include block Toeplitz matrices, which further include the MLT ones, the decomposition derived in \cite{georgiou2007caratheodory} is less structured than and does not imply the decomposition given in this paper. A similar decomposition of block Toeplitz matrices as in \cite{georgiou2007caratheodory} was also introduced in \cite{gurvits2002largest} which, as we will see, is useful for deriving the result of this paper. A recent paper \cite{chi2015compressive} attempted to generalize the Vandermonde decomposition to the 2D case and provided a result similar to Theorem \ref{thm:vanderdec} in the present paper; however, its proof is incomplete and some derivations are flawed. In fact, its proof is almost identical to that in \cite{gurvits2002largest} for block Toeplitz matrices (see Remark \ref{rem:chi}).

Toeplitz matrices can be viewed as discrete counterparts of Toeplitz operators that together with Hankel ones form an important class in operator theory. In the 1D case, Kronecker discovered in the nineteenth century a Vandermonde-like decomposition of Hankel matrices that holds in general with the exception of degenerate cases \cite{kronecker1895leopold,rochberg1987toeplitz}. The Vandermonde decomposition by Carath\'{e}odory and Fej\'{e}r can be viewed as a more precise result of Kronecker's Theorem obtained by imposing the condition of PSDness that completely avoids the degenerate cases. A recent paper \cite{andersson2015general} studied the multivariable Hankel operators and showed that Kronecker's Theorem in the MD case differs from its 1D counterpart in several key aspects. In particular, an ML Hankel matrix often does not admit a Vandermonde-like decomposition. In contrast, we show in this paper that by imposing the PSDness all the MLT matrices admit a Vandermonde decomposition under an appropriate rank condition. In this context we note that it would be of interest to investigate the continuous counterpart of the result of this paper for multivariable Toeplitz operators.

The Vandermonde decomposition of MLT matrices provides the theoretical basis of several subspace methods for 2D and higher-dimensional frequency estimation proposed in the 1990s, which typically estimate the frequencies from an estimate of the MLT covariance matrix (see, e.g., \cite{hua1993pencil,haardt19952d,li1992two,hua1992estimating}). The super-resolution methods presented in this paper are inspired by \cite{yang2014enhancing} and also can be viewed as covariance-based methods similarly to the subspace methods (see also \cite{yang2014exact, yang2015gridless}). But the main difference is that in this paper the covariance estimates are obtained by optimization of sophisticated covariance-fitting criteria that fully exploit the MLT structure and utilize signal sparsity. Moreover, the used criteria work in the presence of missing data. Note that 2D frequency estimation has also been studied in \cite{chi2015compressive,chen2014robust,andersson2014new}. The paper \cite{chi2015compressive} analyzed the performance of the atomic norm method, which generalizes the result in \cite{tang2012compressed} from the 1D to the 2D case. Both \cite{chen2014robust} and \cite{andersson2014new} exploited the low-rankness of a certain 2-level Hankel matrix formed using the data samples. Since the methods in these papers actually can be applied to the case of general complex exponentials, their use for the frequency estimation of sinusoids appears to produce suboptimal results. Moreover, as mentioned above a Vandermonde-like decomposition may not exist for an ML Hankel matrix \cite{andersson2015general} and therefore the use of the decomposition for parameter retrieval can be problematic.

\subsection{Notations}
Notations used in this paper are as follows. $\bR$, $\bC$ and $\bZ$ denote the sets of real, complex and integer numbers respectively. $\bT$ denotes the unit circle $\left[0,1\right]$ by identifying the beginning and the ending points. Boldface letters are reserved for vectors and matrices. $\cdot^T$ and $\cdot^H$ denote the matrix transpose and the Hermitian transpose. $\onen{\cdot}$ and $\twon{\cdot}$ represent the $\ell_1$ and $\ell_2$ norms. For two vectors $\m{a}$ and $\m{b}$, $\m{a}\leq \m{b}$ is understood elementwise. For two square matrices $\m{A}$ and $\m{B}$, $\m{A}\geq \m{B}$ means that $\m{A}-\m{B}$ is positive semidefinite.


\subsection{Organization of This Paper}
The rest of the paper is organized as follows. Section \ref{sec:notation} introduces some preliminaries. Section \ref{sec:genVanDec} presents the main contribution---the Vandermonde decomposition of MLT matrices---as well as the methods for finding the decomposition. In Section \ref{sec:super-resolution} the obtained results are applied to studying the MD super-resolution problem using an atomic $\ell_0$ norm method. In Section \ref{sec:simulation} extensive numerical simulations are provided to validate the theoretical findings and demonstrate the performance of the proposed super-resolution approaches. Section \ref{sec:conclusion} concludes this paper.

\section{Preliminaries} \label{sec:notation}

\subsection{Toeplitz and MLT Matrices}
Given a complex sequence $\m{u}=\mbra{u_k}$, $k\in\bZ$, an $n\times n$ Toeplitz matrix $\m{T}_{n}$ is defined as
\equ{\m{T}_{n} \coloneqq \begin{bmatrix} u_0 & u_{1} & \dots & u_{n-1} \\ u_{-1} & u_0 & \dots & u_{n-2} \\ \vdots & \vdots & \ddots & \vdots \\ u_{1-n} & u_{2-n} & \dots & u_0 \end{bmatrix}. \label{eq:toep} }
For $d\geq2$, let $\m{n} = \sbra{n_1,\dots,n_d}$ and $\m{n}_{-1}= \sbra{n_2,\dots,n_d}$. Given a $d$-dimensional ($d$D) complex sequence $\m{u}=\mbra{u_{\m{k}}}$, $\m{k}\in\bZ^d$, an $\m{n}$, $d$-level Toeplitz ($d$LT) matrix $\m{T}_{\m{n}}$ is defined recursively as:
\equ{\m{T}_{\m{n}}\coloneqq \begin{bmatrix} \m{T}_{0\m{n}_{-1}} & \m{T}_{1\m{n}_{-1}} & \dots & \m{T}_{\sbra{n_1-1}\m{n}_{-1}} \\ \m{T}_{\sbra{-1}\m{n}_{-1}} & \m{T}_{0\m{n}_{-1}} & \dots & \m{T}_{\sbra{n_1-2}\m{n}_{-1}} \\ \vdots & \vdots & \ddots & \vdots \\ \m{T}_{\sbra{1-n_1}\m{n}_{-1}} & \m{T}_{\sbra{2-n_1}\m{n}_{-1}} & \dots & \m{T}_{0\m{n}_{-1}} \end{bmatrix}, \label{eq:MLT}}
where for $k_1=1-n_1,\dots,n_1-1$, $\m{T}_{k_1\m{n}_{-1}}$ denotes an $\m{n}_{-1}$, $\sbra{d-1}$LT matrix formed using $\mbra{u_{\m{k}}}$, $-\m{n}_{-1}\leq \m{k}_{-1}\leq \m{n}_1$. It can be seen from \eqref{eq:MLT} that $\m{T}_{\m{n}}$ is an $n_1\times n_1$ block Toeplitz matrix in which each block is an $\m{n}_{-1}$, $\sbra{d-1}$LT matrix and thus $\m{T}_{\m{n}}\in\bC^{N\times N}$, where $N=\prod_{l=1}^d n_l$. As an example, in the case of $\m{n}=(2,2)$ we have that
\equ{\m{T}_{\sbra{2,2}} = \begin{bmatrix} u_{00} & u_{01} & u_{10} & u_{11} \\ u_{0(-1)} & u_{00} & u_{1(-1)} & u_{10} \\ u_{(-1)0} & u_{(-1)1} & u_{00} & u_{01} \\ u_{(-1)(-1)} & u_{(-1)0} & u_{0(-1)} & u_{00} \end{bmatrix}. }
Note that a 2LT matrix is also called Toeplitz-block-Toeplitz or doubly Toeplitz in the literature. For notational simplicity, we will omit the index $\m{n}$ in $\m{T}_{\m{n}}$ when it is obvious from the context.

\subsection{Vandermonde Decomposition and Frequency Estimation}

The Vandermonde decomposition of Toeplitz matrices is the basis of subspace methods for 1D frequency estimation. To be specific, let $\m{a}_n\sbra{f}=n^{-\frac{1}{2}}\mbra{1,e^{i2\pi f},\dots, e^{i2\pi (n-1)f}}^T\in\bC^{n}$ denote a uniformly sampled complex sinusoid with frequency $f\in\bT$ and unit power, where $i=\sqrt{-1}$. It follows that for $\m{f}\in\bT^r$, $\m{A}_{n}\sbra{\m{f}} \coloneqq \mbra{\m{a}_n\sbra{f_{1}}, \dots, \m{a}_n\sbra{f_{r}}}\in\bC^{n\times r}$ is a Vandermonde matrix (up to a factor of $n^{-\frac{1}{2}}$). Let us consider the parametric model for frequency estimation
\equ{\m{y} = \m{A}_n\sbra{\m{f}}\m{c} = \sum_{j=1}^r c_j \m{a}_n\sbra{f_j}, \label{eq:paramodel}}
where $c_j=\abs{c_j}e^{i\phi_j}\in\bC$ are complex amplitudes, $\phi_j$ are initial phases, and $\m{y}\in\bC^n$ denotes the sampled data. Assume that the sinusoids have i.i.d. random initial phases. It follows that $\m{P}\coloneqq \bE \m{c}\m{c}^H = \diag\sbra{\abs{c_1}^2,\dots,\abs{c_r}^2}$. Then, the data covariance matrix
\equ{\m{R} = \bE \m{y}\m{y}^H = \m{A}_n\sbra{\m{f}}\m{P}\m{A}_n^H\sbra{\m{f}} = \sum_{j=1}^r \abs{c_j}^2 \m{a}_n\sbra{f_j}\m{a}_n^H\sbra{f_j} \label{eq:Vanderdecmp1}}
is a rank-$r$ PSD Toeplitz matrix (assuming that $r< n$ and $f_j$, $j=1,\dots,r$ are distinct). The sequence $\m{u}$ used to generate the Toeplitz matrix $\m{R}$ is given by
\equ{u_k = \sum_{j=1}^r \abs{c_j}^2 e^{-i2\pi kf_j}, \quad 1-n\leq k\leq n-1. }
The Vandermonde decomposition states that the converse is also true. That is, any PSD Toeplitz matrix $\m{R}$ of rank $r<n$ can always be uniquely decomposed as in \eqref{eq:Vanderdecmp1}. Consequently, the frequencies can be unambiguously retrieved from the data covariance [note that in the presence of white noise, the noise contribution to $\m{R}$ can also be identified (see \cite{pisarenko1973retrieval})]. In practice, $\m{R}$ can only be approximately estimated and subspace methods like MUSIC and ESPRIT have been proposed to carry out the frequency estimation task.

In the MD case, let $\m{f}\in\bT^{d\times r}$ denote a set of $r$, $d$D frequencies $\m{f}_{:j}\in\bT^d$, $j=1,\dots,r$, where $\m{f}_{:j}$ can be understood as the $j$th column of $\m{f}$. Let $\m{f}_{l}$ denote the $l$th row of $\m{f}$ or the set of frequencies for the $l$th dimension, and $f_{lj}$ be the $\sbra{l,j}$th entry. A uniformly sampled $d$D complex sinusoid with frequency $\m{f}_{:j}$ and unit power can be represented by
$\m{a}_{\m{n}}\sbra{\m{f}_{:j}} \coloneqq \m{a}_{n_1}\sbra{f_{1j}}\otimes \dots \otimes \m{a}_{n_d}\sbra{f_{dj}}\in\bC^{N}$, where $\otimes$ denotes the Kronecker product and the index $\m{n}$ implies that the sample size is $n_l$ along the $l$th dimension. It follows that
$\m{A}_{\m{n}}\sbra{\m{f}} \coloneqq \mbra{\m{a}_{\m{n}}\sbra{\m{f}_{:1}}, \dots, \m{a}_{\m{n}}\sbra{\m{f}_{:r}}} = \m{A}_{n_1}\sbra{\m{f}_{1}}\star \dots \star \m{A}_{n_d}\sbra{\m{f}_{d}}\in\bC^{N\times r}$, where $\star$ denotes the Khatri-Rao product (or column-wise Kronecker product). Due to the fact that $\m{A}_{\m{n}}\sbra{\m{f}}$ can be written as the Khatri-Rao product of $d$ Vandermonde matrices, we still call $\m{A}_{\m{n}}\sbra{\m{f}}$ a Vandermonde matrix.
In the problem of MD frequency estimation, the sampled data $\m{y}$ follows a similar parametric model:
\equ{\m{y} = \m{A}_{\m{n}}\sbra{\m{f}}\m{c} = \sum_{j=1}^r c_j \m{a}_{\m{n}}\sbra{\m{f}_{:j}}. \label{eq:paramodel2}}
Under the same assumption on the initial phases, the data covariance matrix
\equ{\begin{split}\m{R}= \bE \m{y}\m{y}^H
&= \m{A}_{\m{n}}\sbra{\m{f}} \m{P}\m{A}_{\m{n}}^H\sbra{\m{f}}\\
&= \sum_{j=1}^r \abs{c_j}^2 \m{a}_{\m{n}}\sbra{\m{f}_{:j}}\m{a}_{\m{n}}^H\sbra{\m{f}_{:j}}\\
&= \sum_{j=1}^r \abs{c_j}^2 \bigotimes_{l=1}^d \m{a}_{n_l}\sbra{f_{lj}} \m{a}_{n_l}^H\sbra{f_{lj}} \label{eq:Vanderdecmp2} \end{split}}
turns out to be a PSD $\m{n}$, $d$LT matrix of rank no greater than $r$ and is generated by the sequence
\equ{u_{\m{k}} = \sum_{j=1}^r \abs{c_j}^2 e^{-i2\pi \m{k}^T \m{f}_{:j}},\quad -\m{n}\leq \m{k}\leq \m{n}.}
The fundamental question as to whether also the converse holds true has remained unresolved, though numerical tools such as extensions of MUSIC and ESPRIT have been developed for 2D frequency estimation from a data covariance estimate.


\section{Vandermonde Decomposition of MLT Matrices} \label{sec:genVanDec}

\subsection{Generalizing the Vandermonde Decomposition} \label{sec:VanDec}

A main contribution of this paper is summarized in the following theorem which generalizes the Carath\'{e}odory-Fej\'{e}r's result from the 1D to the MD case (note that we will omit the index $\m{n}$ in $\m{T}_{\m{n}}$, $\m{a}_{\m{n}}$ and $\m{A}_{\m{n}}$ for simplicity).

\begin{thm} Assume that $\m{T}$ is a PSD $d$LT matrix with $d\geq1$ and $\rank\sbra{\m{T}}=r< \min_j n_j$. Then, $\m{T}$ can be decomposed as
\equ{\m{T} = \m{A}\sbra{\m{f}}\m{P}\m{A}\sbra{\m{f}} = \sum_{j=1}^r p_j \m{a}\sbra{\m{f}_{:j}} \m{a}^H\sbra{\m{f}_{:j}}, \label{eq:vanderdec}}
where $\m{P}=\diag\sbra{p_1,\dots,p_r}$ with $p_j>0$, $\m{f}_{:j}$, $j=1,\dots,r$ are distinct points in $\bT^d$, and the $(d+1)$-tuples $\sbra{\m{f}_{:j},p_j}$, $j=1,\dots,r$ are unique. \label{thm:vanderdec}
\end{thm}

\begin{rem} The sufficient condition that $\rank\sbra{\m{T}}=r< \min_j n_j$ of Theorem \ref{thm:vanderdec} is tight. Indeed, for $r\geq \min_j n_j$ we can always find $\m{T}$ matrices that admit infinitely many Vandermonde decompositions of order $r$. Consider the case of $r= \min_j n_j$ as an example. For $d=1$, first note that $\m{T}$ is invertible. For any $f_1\in\bT$, let $p_1=\mbra{\m{a}_1^H\sbra{f_1}\m{T}^{-1}\m{a}_1\sbra{f_1}}^{-1}$. Then it is easy to see that $\m{T} - p_1\m{a}_1\sbra{f_1}\m{a}_1^H\sbra{f_1}$ is a PSD Toeplitz matrix of rank $r-1$ and thus it has a unique Vandermonde decomposition of order $r-1$. This means that we have found a Vandermonde decomposition of order $r$ for $\m{T}$. Since $f_1$ above can be chosen arbitrarily, there exist infinitely many such decompositions. For $d\geq2$, without loss of generality, assume that $n_1= r$ and $n_2,\dots,n_d\geq r$. Then for any PSD Toeplitz matrices $\m{T}_{n_l}\in\bC^{n_l\times n_l}$ with $\rank\sbra{\m{T}_{n_1}}=r$ and $\rank\sbra{\m{T}_{n_l}}=1$, $l=2,\dots,d$, the $d$LT matrix $\m{T}=\bigotimes_{l=1}^d \m{T}_{n_l}$ admits infinitely many Vandermonde decompositions of order $r$ because $\m{T}_{n_1}$ does so (as shown previously). \label{rem:tight}
\end{rem}

\begin{rem} For $d=2$ a similar result to Theorem \ref{thm:vanderdec} was recently presented in \cite[Proposition 2]{chi2015compressive}; however, its proof is incomplete and certain derivations are flawed. In particular, the main part of the proof in \cite{chi2015compressive} is nothing but the first step of ours that follows from \cite{gurvits2002largest} and holds for general block Toeplitz matrices (see below). Moreover, Eq. (44) in \cite{chi2015compressive}, which provides a Vandermonde decomposition of $\m{T}$ and concludes Proposition 2 in \cite{chi2015compressive}, does not hold true. To see that, consider the case where $\lbra{f_{2j}}_{j=1}^r$ have identical entries. Then, the $\lbra{f_{2j}}_{j=1}^r$ constructed in \cite[Eq. (40)]{chi2015compressive} are not unique (note that a typo exists in \cite[Eq. (40)]{chi2015compressive} where $f_{1i}$ should be $f_{2i}$). It follows that the Vandermonde decomposition constructed in \cite[Eq. (44)]{chi2015compressive} is not unique either, which cannot be true according to Theorem \ref{thm:vanderdec} of this paper. \label{rem:chi}
\end{rem}

To prove Theorem \ref{thm:vanderdec}, we first consider the uniqueness part which essentially follows from the following lemma.

\begin{lem} Assume that $\lbra{f_{1j}}_{j=1}^{n_1}$, \dots, $\lbra{f_{dj}}_{j=1}^{n_d}$ are $d\geq 1$ sets of distinct points in $\bT$. Then,
\equ{\lbra{\m{a}\sbra{f_{1j_1},\dots,f_{dj_d}}:\; j_l = 1,\dots,n_l,\; l=1,\dots,d} \label{eq:colvec}}
are linearly independent. \label{lem:linindep}
\end{lem}


\begin{proof} For $d=1$ the result is well known and its proof is therefore omitted. It follows that $\m{A}_{n_l}\sbra{\m{f}_l}$, $l=1,\dots,d$ are all invertible. For $d\geq 2$, note that the vectors in \eqref{eq:colvec} form the $N\times N$ matrix $\bigotimes_{l=1}^d \m{A}_{n_l}\sbra{\m{f}_l}$. We complete the proof by the fact that
\equ{\rank\sbra{\bigotimes_{l=1}^d \m{A}_{n_l}\sbra{\m{f}_l}} = \prod_{l=1}^d \rank\sbra{\m{A}_{n_l}\sbra{\m{f}_l}} = \prod_{l=1}^d n_l= N.}
%
\end{proof}

The existence of the Vandermonde decomposition is proven via a constructive method. The proof is motivated by the decomposition of block Toeplitz matrices given in \cite{gurvits2002largest}, the proof of which is also provided for completeness.

\begin{lem}[\cite{gurvits2002largest}] Let $\m{T}^B\in\bC^{mn\times mn}$ be an $n\times n$ PSD block Toeplitz matrix with $\rank\sbra{\m{T}^B}=r$. Then there exist $\m{V}=\mbra{\dots,\m{v}_j,\dots}\in\bC^{m\times r}$ and $f_j\in\bT$, $j=1,\dots,r$ such that $\m{T}^B$ can be decomposed as
\equ{\begin{split}\m{T}^B
&= \sum_{j=1}^r \m{a}_{n}\sbra{f_{j}}\m{a}_{n}^H\sbra{f_{j}} \otimes \m{v}_j\m{v}_j^H \\
&= \sum_{j=1}^r \sbra{\m{a}_{n}\sbra{f_{j}}\otimes \m{v}_j} \sbra{\m{a}_{n}\sbra{f_{j}}\otimes \m{v}_j}^H. \end{split} \label{eq:Tninvj0}} \label{lem:Tl}
\end{lem}
\begin{proof} Since $\m{T}^B$ is PSD with $\rank\sbra{\m{T}^B}=r$, there exists $\m{Y}\in\bC^{mn\times r}$ such that $\m{T}^B=\m{Y}\m{Y}^H$. Write $\m{Y}$ as $\m{Y}=\mbra{\m{Y}_0^H,\dots,\m{Y}_{n-1}^H}^H$ with $\m{Y}_j\in\bC^{m\times r}$, $j=0,\dots,n-1$. Define the upper submatrix $\m{Y}_U = \mbra{\m{Y}_0^H,\dots,\m{Y}_{n-2}^H}^H$ and the lower submatrix $\m{Y}_L = \mbra{\m{Y}_1^H,\dots,\m{Y}_{n-1}^H}^H$. By the block Toeplitz structure of $\m{T}^B$ it holds that
\equ{\m{Y}_U\m{Y}_U^H=\m{Y}_L\m{Y}_L^H.}
Thus there exists a unitary matrix $\m{U}\in\bC^{r\times r}$ such that $\m{Y}_L = \m{Y}_U\m{U}$ (see, e.g., \cite[Theorem 7.3.11]{horn2012matrix}). It follows that
\equ{\m{Y}=\mbra{\m{Y}_0^H,\sbra{\m{Y}_0\m{U}}^H\dots,\sbra{\m{Y}_0\m{U}^{n-1}}^H}^H.}
Let $\m{T}_l$, $l=1-n,\dots,n-1$ denote the matrix on the $l$th block diagonal of $\m{T}^B$. So we have that
\equ{\m{T}_l = \m{Y}_0\m{U}^{-l}\m{Y}_0^H, \quad l=1-n,\dots,n-1. \label{eq:TlYUY}}
Next, write the eigen-decomposition of the unitary matrix $\m{U}$, which is guaranteed to exist, as
\equ{\m{U}=\widetilde{\m{U}}\m{Z}\widetilde{\m{U}}^H, \label{eq:eigdecU}}
where $\m{Z}=\diag\sbra{\dots,z_j,\dots}$ and $\widetilde{\m{U}}$ is another unitary matrix. The eigenvalues $z_j$, $j=1,\dots,r$ have magnitude of 1 and thus $z_j=e^{i2\pi f_j}$, $f_j\in\bT$. Inserting \eqref{eq:eigdecU} into \eqref{eq:TlYUY} and letting $\m{V}=\m{Y}_0\widetilde{\m{U}}$, we have that
\equ{\m{T}_l =\m{V}\m{Z}^{-l}\m{V}^H= \sum_{j=1}^r e^{-i2\pi l f_{j}} \m{v}_j\m{v}_j^H, \label{eq:Tl2}}
where $\m{v}_j$ denotes the $j$th column of $\m{V}$. Finally, \eqref{eq:Tninvj0} follows from \eqref{eq:Tl2}.
\end{proof}

To derive the Vandermonde decomposition based on Lemma \ref{lem:Tl}, a key result that we will use is the following.

\begin{lem} If a $d$LT matrix $\m{T}$, $d\geq1$, can be written as
\equ{\m{T} = \m{A}\sbra{\m{f}} \m{C} \m{A}^H\sbra{\m{f}}, \label{eq:TBCBH}}
where $\m{C}\in\bC^{r\times r}$, $r< \min_j n_j$ and $\m{f}_{:j}$, $j=1,\dots,r$ are distinct points in $\bT^d$, then $\m{C}$ must be a diagonal matrix. \label{lem:Tdiag}
\end{lem}

The proof of Lemma \ref{lem:Tdiag} is somewhat complicated and thus it is deferred to Appendix \ref{append:Tdiag}. Note that we first prove the result in the case of $d=1$ using the Kronecker's Theorem for Hankel matrices \cite{rochberg1987toeplitz, ellis1992factorization} and the connection between Toeplitz and Hankel matrices. We then complete the proof for $d\geq2$ using induction.

{\it Proof of Theorem \ref{thm:vanderdec}}: We first show that the Vandermonde decomposition in \eqref{eq:vanderdec}, if it exists, is unique.   To do so, suppose that $\m{T}$ admits another decomposition
\equ{\m{T} = \m{A}\sbra{\m{f}'} \m{P}' \m{A}^H\sbra{\m{f}'}, \label{eq:vanderdec2}}
where, as in \eqref{eq:vanderdec}, $\m{P}'=\diag\sbra{p'_1,\dots,p'_r}$ with $p'_k>0$, and $\m{f}'_{:k}$, $k=1,\dots,r$ are distinct points in $\bT^d$. It follows that $\m{A}\sbra{\m{f}'} \m{P}' \m{A}^H\sbra{\m{f}'}= \m{A}\sbra{\m{f}} \m{P} \m{A}^H\sbra{\m{f}}$ and thus $\m{A}\sbra{\m{f}'} \m{P}'^{\frac{1}{2}} = \m{A}\sbra{\m{f}} \m{P}^{\frac{1}{2}}\m{U}'$, where $\m{U}'$ is an $r\times r$ unitary matrix \cite[Theorem 7.3.11]{horn2012matrix}. So we have that
\equ{\m{A}\sbra{\m{f}'} = \m{A}\sbra{\m{f}} \m{P}^{\frac{1}{2}}\m{U}'\m{P}'^{-\frac{1}{2}}. }
This means that each $\m{a}\sbra{\m{f}'_{:k}}$ is a linear combination of $\lbra{\m{a}\sbra{\m{f}_{:j}}}_{j=1}^r$. In other words, for $k=1,\dots,r$, the $r+1\leq \min_j n_j$ vectors $\lbra{\m{a}\sbra{\m{f}'_{:k}}, \m{a}\sbra{\m{f}_{:1}}, \dots, \m{a}\sbra{\m{f}_{:r}}}$
are linearly dependent. By Lemma \ref{lem:linindep} this can be true only if $\m{f}'_{:k}\in \lbra{\m{f}_{:j}}_{j=1}^r$, implying that $\lbra{\m{f}'_{:k}}_{k=1}^r\subset \lbra{\m{f}_{:j}}_{j=1}^r$. By a similar argument, we can also show that $\lbra{\m{f}_{:j}}_{j=1}^r\subset \lbra{\m{f}'_{:k}}_{k=1}^r$. Consequently, $\lbra{\m{f}'_{:j}}_{j=1}^r$ and $\lbra{\m{f}_{:j}}_{j=1}^r$ are identical. Then it follows that the coefficients $\lbra{p_j}$ and $\lbra{p'_j}$ are identical as well.

We use induction to prove the existence part. First of all, for $d=1$ the result turns out to be the standard Vandermonde decomposition of Toeplitz matrices and thus it holds true. Suppose that a decomposition as in \eqref{eq:vanderdec} exists for $d=d_0-1$, $d_0\geq 2$. It suffices to prove that it also exists for $d=d_0$. We complete the proof in three steps.
In \emph{Step 1}, by viewing $\m{T}$ as an $n_1\times n_1$ block Toeplitz matrix and applying Lemma \ref{lem:Tl}, we have that
\equ{\begin{split}\m{T}
&= \sum_{j=1}^r \m{a}_{n_1}\sbra{f_{1j}}\m{a}_{n_1}^H\sbra{f_{1j}} \otimes \m{v}_j\m{v}_j^H \\
&= \sum_{j=1}^r \sbra{\m{a}_{n_1}\sbra{f_{1j}}\otimes \m{v}_j} \sbra{\m{a}_{n_1}\sbra{f_{1j}}\otimes \m{v}_j}^H, \end{split} \label{eq:Tninvj}}
where $f_{1j}\in\bT$, $j=1,\dots,r$. The following identity that follows from \eqref{eq:Tl2} will also be used later:
\equ{\m{T}_{-l} =\m{V}\m{Z}_1^{l}\m{V}^H, \quad l=0,\dots,n_1-1, \label{eq:Tl22}}
where $\m{Z}_1=\diag\sbra{e^{i2\pi f_{11}},\dots, e^{i2\pi f_{1r}}}$.

In \emph{Step 2}, we consider the first block of $\m{T}$, $\m{T}_0 = \m{V}\m{V}^H$ that is a PSD $\sbra{d_0-1}$LT matrix. Let $r'=\rank\sbra{\m{T}_0}\leq r<\min_j n_j$. By the assumption that a Vandermonde decomposition exists for $d=d_0-1$, $\m{T}_0$ admits the following Vandermonde decomposition:
\equ{\begin{split}\m{T}_0
&= \m{A}_{\m{n}_{-1}}\sbra{\widetilde{\m{f}}_{-1}} \widetilde{\m{P}} \m{A}_{\m{n}_{-1}}^H\sbra{\widetilde{\m{f}}_{-1}}\\
&= \sum_{j=1}^{r'} \widetilde{p}_j \m{a}_{\m{n}_{-1}}\sbra{\widetilde{\m{f}}_{-1,j}} \m{a}_{\m{n}_{-1}}^H\sbra{\widetilde{\m{f}}_{-1,j}}, \end{split} \label{eq:VDT0}}
where $\widetilde{\m{f}}_{-1,j}$ is the $j$th column in $\widetilde{\m{f}}_{-1}\in\bT^{\sbra{d_0-1}\times r}$, $\widetilde{\m{f}}_{-1,j}$, $j=1,\dots,r'$ are distinct, and $\widetilde{\m{P}} = \diag\sbra{\widetilde{p}_1,\dots,\widetilde{p}_{r'}}$ with $\widetilde{p}_j>0$, $j=1,\dots,r'$. Because $\m{T}_0=\m{V}\m{V}^H = \m{A}_{\m{n}_{-1}}\sbra{\widetilde{\m{f}}_{-1}} \widetilde{\m{P}} \m{A}_{\m{n}_{-1}}^H\sbra{\widetilde{\m{f}}_{-1}}$, it holds that \cite[Theorem 7.3.11]{horn2012matrix}
\equ{\m{V} = \m{A}_{\m{n}_{-1}}\sbra{\widetilde{\m{f}}_{-1}}  \widetilde{\m{P}}^{\frac{1}{2}} \m{O}, \label{eq:VinA2}}
where $\m{O}\in\bC^{r'\times r}$ and $\m{O}\m{O}^H = \m{I}$. Inserting \eqref{eq:VinA2} into \eqref{eq:Tl22}, we have that
\equ{\begin{split}\m{T}_{-l}
=& \m{A}_{\m{n}_{-1}}\sbra{\widetilde{\m{f}}_{-1}} \widetilde{\m{P}}^{\frac{1}{2}} \m{O} \m{Z}_1^{l} \m{O}^H \widetilde{\m{P}}^{\frac{1}{2}} \m{A}_{\m{n}_{-1}}^H\sbra{\widetilde{\m{f}}_{-1}} ,\\
& l=0,\dots,n_1-1. \end{split} \label{eq:Tl3}}
It immediately follows from Lemma \ref{lem:Tdiag} that $\widetilde{\m{P}}^{\frac{1}{2}} \m{O} \m{Z}_1^l \m{O}^H \widetilde{\m{P}}^{\frac{1}{2}}$, $l=0,1,\dots, n_1-1$ are diagonal matrices and so are $\m{O} \m{Z}_1^l \m{O}^H$.

In \emph{Step 3}, we show that $\m{O}$ and $\m{V}$ are structured, which together with \eqref{eq:Tninvj} leads to a decomposition of $\m{T}$ as in \eqref{eq:vanderdec}. To do so, let
\equ{\m{D}\sbra{l} \coloneqq \m{O} \m{Z}_1^l \m{O}^H = \sum_{j=1}^r e^{i2\pi l f_{1j}} \m{O}\sbra{j}, \quad l=0,1\dots,n_1-1 \label{eq:Gl}}
be a series of diagonal matrices, where $\m{O}\sbra{j} = \m{o}_j\m{o}_j^H$ and $\m{o}_j$ is the $j$th column of $\m{O}$.
First consider the case when $f_{1j}$, $j=1,\dots,r$ are distinct. For the $\sbra{m,n}$th entry of $\m{D}\sbra{l}$, denoted by $D_{mn}\sbra{l}$, we have the following linear system of equations whenever $m\neq n$:
\equ{\begin{bmatrix} 0\\ 0 \\ \vdots \\ 0 \end{bmatrix} = \begin{bmatrix} D_{mn}(0)\\ D_{mn}(1) \\ \vdots \\ D_{mn}(n_1-1) \end{bmatrix} = \m{A}_{n_1}\sbra{\m{f}_1} \begin{bmatrix} O_{mn}(1)\\ O_{mn}(2) \\ \vdots \\ O_{mn}(r) \end{bmatrix}, \label{eq:Gmnis0}}
where $\m{A}_{n_1}\sbra{\m{f}_1}$ has full column rank since $r< n_1$. It immediately follows that $O_{mn}\sbra{j}=0$, $j=1,\dots,r$ when $m\neq n$, i.e., $\m{O}\sbra{j}$ are diagonal matrices. Moreover, each $\m{O}\sbra{j}$ contains at most one nonzero entry on its diagonal since its rank is at most 1. This means that $\m{o}_j$ has at most one nonzero entry and hence, $\m{v}_j = \m{A}_{\m{n}_{-1}}\sbra{\widetilde{\m{f}}_{-1}} \widetilde{\m{P}}^{\frac{1}{2}} \m{o}_j$ is the product of a scalar and some column in $\m{A}_{\m{n}_{-1}}\sbra{\widetilde{\m{f}}_{-1}} $. As a result, we obtain from \eqref{eq:Tninvj} a decomposition of $\m{T}$ as in \eqref{eq:vanderdec}.

We next consider the other case when some $f_{1j}$'s are identical. Without loss of generality, we assume that $f_{1j}$, $j=1,\dots, r_0\leq r$ are identical and different from the others. By similar arguments we can conclude that $\sum_{j=1}^{r_0} \m{o}_j\m{o}_j^H$ is a diagonal matrix of rank at most $r_0$. Then,
\equ{\begin{split}
&\sum_{j=1}^{r_0} \m{v}_j\m{v}_j^H\\
&= \m{A}_{\m{n}_{-1}}\sbra{\widetilde{\m{f}}_{-1}} \widetilde{\m{P}}^{\frac{1}{2}} \sbra{\sum_{j=1}^{r_0} \m{o}_j\m{o}_j^H} \widetilde{\m{P}}^{\frac{1}{2}} \m{A}_{\m{n}_{-1}}^H\sbra{\widetilde{\m{f}}_{-1}} \\
&= \sum_{j=1}^{r_0} p_j \m{a}_{\m{n}_{-1}}\sbra{\m{f}_{-1,j}}\m{a}_{\m{n}_{-1}}^H\sbra{\m{f}_{-1,j}}, \end{split} \label{eq:sumv}}
where $p_j\geq0$ and $\m{f}_{-1,j} \in \lbra{\widetilde{\m{f}}_{-1,1},\dots,\widetilde{\m{f}}_{-1,r'}}$ for $j=1,\dots,r_0$. Inserting \eqref{eq:sumv} into \eqref{eq:Tninvj}, we have that
\equ{\begin{split}\m{T}
&= \m{a}_{n_1}\sbra{f_{11}}\m{a}_{n_1}^H\sbra{f_{11}} \otimes \sbra{\sum_{j=1}^{r_0} \m{v}_j\m{v}_j^H} \\
&\quad + \sum_{j=r_0+1}^r \m{a}_{n_1}\sbra{f_{1j}}\m{a}_{n_1}^H\sbra{f_{1j}} \otimes \m{v}_j\m{v}_j^H \\
&= \sum_{j=1}^{r_0} p_j \m{a}\sbra{f_{11}, \m{f}_{-1,j}} \m{a}^H\sbra{f_{11}, \m{f}_{-1,j}} \\
&\quad + \sum_{j=r_0+1}^r \sbra{\m{a}_{n_1}\sbra{f_{1j}}\otimes \m{v}_j} \sbra{\m{a}_{n_1}\sbra{f_{1j}}\otimes \m{v}_j}^H. \end{split} }
Therefore, by similarly dealing with $f_{1j}$, $j=r_0+1,\dots,r$ we can obtain a decomposition of $\m{T}$ as in \eqref{eq:vanderdec}.
\BOX

%


\subsection{Finding the Vandermonde Decomposition} \label{sec:MaPP}
The proof of Theorem \ref{thm:vanderdec} provides a constructive method for finding the Vandermonde decomposition of $d$LT matrices. In the case of $d=1$, the result becomes the conventional Vandermonde decomposition which can be computed using the algorithm in \cite{pisarenko1973retrieval} or subspace methods (or the matrix pencil method introduced later). For $d\geq2$, by viewing $\m{T}$ as an $n_1\times n_1$ block Toeplitz matrix, we can first obtain a decomposition as in \eqref{eq:Tninvj} following the proof of Lemma \ref{lem:Tl}. Then it suffices to find the Vandermonde decomposition of the $(d-1)$LT matrix $\m{T}_0$ as in \eqref{eq:VDT0}. The pairing between $f_{1j}$, $j=1,\dots,r$ and $\widetilde{\m{f}}_{-1,j}$, $j=1,\dots,r'=\rank\sbra{\m{T}_0}$ can be automatically accomplished via \eqref{eq:VinA2} in which $\m{O}=\widetilde{\m{P}}^{-\frac{1}{2}} \m{A}_{\m{n}_{-1}}^{\dag}\sbra{\widetilde{\m{f}}_{-1}}\m{V}$, where $\cdot^{\dag}$ denotes the matrix pseudo-inverse. As a result, the Vandermonde decomposition can be computed in a sequential manner, from the 1L to the 2L and finally to the $d$L case.

We next study in detail the computations of $f_{1j}$, $j=1,\dots,r$ and $\m{V}$ in \eqref{eq:Tninvj}. Recalling the proof of Lemma \ref{lem:Tl}, we will use the identities $\m{T}=\m{Y}\m{Y}^H$, $\m{Y}_L = \m{Y}_U\m{U}$, $\m{U} = \widetilde{\m{U}}\m{Z}_1\widetilde{\m{U}}^H$ and $\m{V} = \m{Y}_0\widetilde{\m{U}}$, where $\m{Z}_1=\diag\sbra{z_{11},\dots,z_{1r}}$ with $z_{1j}=e^{i2\pi f_{1j}}$, $j=1,\dots,r$. We consider the matrix pencil $\sbra{\m{Y}_U^H\m{Y}_L, \m{Y}_U^H\m{Y}_U}$. It holds that
\equ{\begin{split}\m{Y}_U^H\m{Y}_L - \lambda \m{Y}_U^H\m{Y}_U
&= \m{Y}_U^H\m{Y}_U \widetilde{\m{U}}\m{Z}_1\widetilde{\m{U}}^H - \lambda \m{Y}_U^H\m{Y}_U \\
&= \m{Y}_U^H\m{Y}_U\widetilde{\m{U}} \sbra{\m{Z}_1 - \lambda\m{I}} \widetilde{\m{U}}^H
\end{split}}
and thus
\equ{\sbra{\m{Y}_U^H\m{Y}_L - z_{1j} \m{Y}_U^H\m{Y}_U } \widetilde{\m{u}}_j = \m{0},\quad j=1,\dots,r,}
where $\widetilde{\m{u}}_j$ denotes the $j$th column of $\widetilde{\m{U}}$. This means that $z_{1j}$ and $\widetilde{\m{u}}_j$, $j=1,\dots,r$ are the eigenvalues and eigenvectors of the matrix pencil $\sbra{\m{Y}_U^H\m{Y}_L, \m{Y}_U^H\m{Y}_U}$ whose computation is a generalized eigenproblem. Following this observation, the algorithm described above for the Vandermonde decomposition is designated as the {\em ma}trix {\em p}encil and (auto-){\em p}airing (MaPP) method. Since the main computational step of MaPP is the factorization $\m{T}=\m{Y}\m{Y}^H$, MaPP requires $O\sbra{N^2r}$ flops.

%

\subsection{The Case of $r\geq \min_j n_j$} \label{Sec:MaPP}
As mentioned in Remark \ref{rem:tight}, the condition that $\rank\sbra{\m{T}}=r< \min_j n_j$ of Theorem \ref{thm:vanderdec} is tight. Even so, it is of interest to study the existence of a Vandermonde decomposition and to find it in the case of $r\geq \min_j n_j$. A numerical algorithm to do this is provided in this subsection. The algorithm is motivated by the following observations that hold true when a decomposition indeed exists:
\begin{enumerate}
  \item Viewing $\m{T}$ as an $n_1\times n_1$ block Toeplitz matrix, $\lbra{f_{1j}}_{j=1}^r$ can be computed using the matrix pencil method described previously.
  \item Let $\m{\cP}_l$, $l=2,\dots,d$ be permutation matrices that are such that
   \equ{\m{\cP}_l\sbra{\bigotimes_{m=1}^d \m{a}_{n_m}\sbra{f'_{m}}} = \m{a}_{n_l}\sbra{f'_{l}} \otimes \sbra{\bigotimes_{m\neq l} \m{a}_{n_m}\sbra{f'_{m}}} \label{eq:Pl}}
   for any $\m{f}'\in\bT^d$. Then,
   \equ{\m{T}^{\sbra{l}}=\m{\cP}_l\m{T}\m{\cP}_l^T \label{eq:Tbral}}
   remains a $d$LT matrix by exchanging the roles of the first and the $l$th dimension of the $d$D frequencies in the Vandermonde decomposition. It follows that $\lbra{f_{lj}}_{j=1}^r$, $l=2,\dots,d$ can be computed up to re-sorting similarly to $\lbra{f_{1j}}_{j=1}^r$.
  \item Let $\m{T} = \sum_{m=1}^r \sigma_m \m{u}_m\m{u}_m^H$ be a truncated eigen-decomposition with $\sigma_{1} \geq \dots \geq \sigma_r >0$. For any $\m{f}_{:j}$, $j=1,\dots,r$ in the Vandermonde decomposition, $\sum_{m=1}^r \twon{\m{u}_m^H \m{a}\sbra{\m{f}_{:j}}}^2 = 1$ is the maximum value of the function $g\sbra{\m{f}'} = \sum_{m=1}^r \twon{\m{u}_m^H \m{a}\sbra{\m{f}'}}^2$, $\m{f}'\in\bT^d$.
\end{enumerate}

The algorithm is implemented as follows. We first compute $\lbra{f_{lj}}_{j=1}^r$, $l=1,\dots,d$ separately from $\m{T}^{\sbra{l}}$, $l=1,\dots,d$ using the matrix pencil method as in the last subsection, where $\m{T}^{\sbra{1}}=\m{T}$. Since the $f_{lj}$ in $\lbra{f_{lj}}_{j=1}^r$, $l=1,\dots,d$ is not necessarily the $f_{lj}$ in the correct $d$-tuple $\m{f}_{:j}\in\bT^d$, we then form the correct $d$-tuples $\m{f}_{:j}$, $j=1,\dots,r$ such that the equation $\sum_{m=1}^r \twon{\m{u}_m^H \m{a}\sbra{\m{f}_{:j}}}^2 = 1$ is satisfied. Finally, the coefficients $p_j$, $j=1,\dots,r$ are obtained using a least-squares method, and a Vandermonde decomposition of $\m{T}$ is found if the residual is zero. The algorithm is still called matrix pencil and pairing (MaPP). Note that when the matrix rank is lower than the dimension of each Toeplitz block, MaPP is guaranteed to find the unique Vandermonde decomposition in which the pairing is done automatically. When the matrix rank is higher, the pairing is done using a search method and MaPP is not guaranteed to find a decomposition. In the latter case, MaPP requires $O\sbra{N^2r + Nr^{d+1}}$ flops. Note that MaPP is similar to the matrix pencil method in \cite{hua1992estimating} for 2D frequency estimation from a Hankel matrix formed using the data samples or from a data covariance estimate.

For $l=1,\dots,d$, let $\m{e}_l\in\lbra{0,1}^d$ be the vector with one at the $l$th entry and zeros elsewhere. This means that $\m{e}_1,\dots, \m{e}_d$ form the canonical basis for $\bR^d$. A theoretical guarantee for MaPP is provided in the following theorem.

\begin{thm} Assume that $\rank\sbra{\m{T}_{\m{n}-\m{e}_1}} = \dots = \rank\sbra{\m{T}_{\m{n}-\m{e}_d}} = \rank\sbra{\m{T}}=r$. Then MaPP is guaranteed to find an order-$r$ Vandermonde decomposition of $\m{T}$ as in \eqref{eq:vanderdec} if it exists. \label{thm:algsucc}
\end{thm}
\begin{proof} Suppose that $\m{T}$ admits an order-$r$ Vandermonde decomposition as in \eqref{eq:vanderdec}. It suffices to show that the sets of frequencies $\lbra{f_{lj}}_{j=1}^r$, $l=1,\dots,d$ are unique and that MaPP can find them. We first consider the computation of $f_{1j}$, $j=1,\dots,r$ using MaPP. Note that $\m{T}_{\m{n}-\m{e}_1}$ is a principal submatrix of the $n_1\times n_1$ block Toeplitz matrix $\m{T}$ obtained by removing the blocks in the last row and column. Following the proof of Lemma \ref{lem:Tl}, suppose that $\m{T}=\m{Y}\m{Y}^H$, $\m{Y}\in\bC^{N\times r}$. The assumption that $\rank\sbra{\m{T}_{\m{n}-\m{e}_1}}=\rank\sbra{\m{T}}=r$ implies that $\m{Y}_U$ has full column rank. It follows that the unitary matrix $\m{U}$ is unique given $\m{Y}$. Moreover, the matrix pencil method is able to find $f_{1j}$, $j=1,\dots,r$.

We next show the uniqueness of $f_{1j}$, $j=1,\dots,r$. Suppose that $\m{T}=\m{Y}'\m{Y}'^H$, $\m{Y}'\in\bC^{N\times r}$. Then there must exist a unitary matrix $\m{U}'$ such that $\m{Y}'=\m{Y}\m{U}'$. It follows that the new matrix pencil
\equ{\m{Y}_U'^H\m{Y}'_L - \lambda \m{Y}_U'^H\m{Y}'_U = \m{U}'^H\sbra{\m{Y}_U^H\m{Y}_L - \lambda \m{Y}_U^H\m{Y}_U} \m{U}' }
has the same eigenvalues as $\m{Y}_U^H\m{Y}_L - \lambda \m{Y}_U^H\m{Y}_U$, which proves the uniqueness.

We now consider the case of $l\geq 2$. Similar to $\m{T}_{\m{n}-\m{e}_1}$, we define the $\sbra{n_l-1,n_1,\dots,n_{l-1}, n_{l+1},\dots, n_d}$, $d$LT matrix $\breve{\m{T}}^{\sbra{l}}$ as a principal submatrix of the $n_l\times n_l$ block Toeplitz matrix $\m{T}^{\sbra{l}}$ in \eqref{eq:Tbral} obtained by removing the blocks in the last row and column. It follows from previous arguments that if $\rank\sbra{\breve{\m{T}}^{\sbra{l}}} = r$, then $\lbra{f_{lj}}_{j=1}^r$ can be uniquely found by MaPP. In fact, this can be readily shown by the fact that $\breve{\m{T}}^{\sbra{l}}$ is identical to $\m{T}_{\m{n}-\m{e}_l}$ up to permutations of rows and columns.
\end{proof}

\begin{rem} It is easy to check that the assumption of Theorem \ref{thm:algsucc} is satisfied under the condition that $r<\min_j n_j$ of Theorem \ref{thm:vanderdec}. Moreover, in the 1D case the assumption of both Theorem \ref{thm:vanderdec} and Theorem \ref{thm:algsucc} turns out to be $N=n_1>r$, as expected.
\end{rem}

The following result suggests that the assumption of Theorem \ref{thm:algsucc} is weak and therefore MaPP can be expected to work for many MLT matrices.

\begin{prop} Assume that
\equ{r\leq N - \frac{N}{\min_l n_l}}
and that $\m{T}$ is given by \eqref{eq:vanderdec} in which $p_j>0$, $j=1,\dots, r$ and the $dr$ frequencies $f_{lj}$, $l=1,\dots,d$, $j=1,\dots, r$ are drawn from a distribution that is continuous with respect to the Lebesgue measure in $\bT^{dr}$. Then, the assumption of Theorem \ref{thm:algsucc} is satisfied almost surely. \label{prop:asshold}
\end{prop}
\begin{proof} Since $\m{T}=\m{A}\sbra{\m{f}}\m{P}\m{A}^H\sbra{\m{f}}$, it is easy to see that $\m{T}_{\m{n}-\m{e}_l} = \m{A}_{\m{n}-\m{e}_l}\sbra{\m{f}}\m{P}\m{A}_{\m{n}-\m{e}_l}^H\sbra{\m{f}}$. By \cite[Proposition 4]{jiang2001almost}, under the assumptions of Proposition \ref{prop:asshold}, it holds almost surely that $\rank\sbra{\m{A}_{\m{n}-\m{e}_1}\sbra{\m{f}}} = \dots = \rank\sbra{\m{A}_{\m{n}-\m{e}_d}\sbra{\m{f}}} = \rank\sbra{\m{A}\sbra{\m{f}}}=r$ since $r\leq \min_l \sbra{N - \frac{N}{n_l}} = N - \frac{N}{\min_l n_l}$. Then the stated result follows by making use of the fact that $\rank\sbra{\m{T}} = \rank\sbra{\m{A}\sbra{\m{f}}}$ and  $\rank\sbra{\m{T}_{\m{n}-\m{e}_l}} = \rank\sbra{\m{A}_{\m{n}-\m{e}_l}\sbra{\m{f}}}$, $l=1,\dots,d$.
\end{proof}

\section{Application to MD Super-Resolution} \label{sec:super-resolution}

\subsection{MD Super-Resolution via Atomic $\ell_0$ Norm}
The concept of super-resolution introduced in \cite{candes2013towards} refers to the recovery of a (1D or MD) frequency spectrum from coarse scale time-domain samples by exploiting signal sparsity. It circumvents the grid mismatch issue of several recent compressed sensing methods by treating the frequencies as continuous (as opposed to quantized) variables. In this paper we tackle the problem using an atomic $\ell_0$ (pseudo-)norm method instead of the existing atomic norm method. The reasons are three-fold. Firstly, the atomic $\ell_0$ norm exploits the sparsity to the greatest extent possible, while the atomic norm is only a convex relaxation. Secondly, the study of atomic norm methods suffers from a resolution limit condition which is not encountered in the analysis of the atomic $\ell_0$ norm approach (see the 1D case analysis in \cite{yang2014exact,yang2014enhancing}). Lastly, we show that a finite-dimensional formulation exists that can exactly characterize the atomic $\ell_0$ norm, whereas parameter tuning remains a challenging task for the atomic norm (see \cite{xu2014precise}).

Consider the parametric model in \eqref{eq:paramodel2}. We are interested in recovering $\m{f}\in\bC^{d\times r}$ given a set of linear measurements of $\m{y}\in\bC^{N}$. Without loss of generality, we assume that the measurements are given by $\m{z}=\m{L}\m{y}\in\bC^M$, where $\m{L}$ denotes a linear operator. Then all possible vectors $\m{y}$ form a convex set $\m{\cC}\coloneqq\lbra{\m{y}:\; \m{z}=\m{L}\m{y}}$. The atomic $\ell_0$ norm of $\m{y}$ is defined as
\equ{\norm{\m{y}}_{\cA,0} \coloneqq \inf_{c_j\in\bC,\, \m{f}_{:j}\in\bT^d}\lbra{K:\; \m{y} = \sum_{j = 1}^K c_j \m{a}\sbra{\m{f}_{:j}} }. }
We propose the following approach for signal and frequency recovery by exploiting signal sparsity:
\equ{\min_{\m{y}} \norm{\m{y}}_{\cA,0}, \st \m{y}\in\m{\cC}. \label{eq:problem}}
Therefore, we estimate $\m{y}$ using the sparsest candidate $\m{y}^*$, which has an atomic decomposition of the minimum order, and in this process we obtain estimates of $\m{f}_{:j}$, $j=1,\dots,r$ using the frequencies in the atomic decomposition of $\m{y}^*$. Regarding theoretical guarantees for the atomic $\ell_0$ norm minimization method, we have the following result.

\begin{prop} Given $\m{z}=\m{L}\m{y}$ where $\m{y}$ is defined in \eqref{eq:paramodel2}, the atomic decomposition of the optimizer of \eqref{eq:problem} exactly recovers the frequencies $\m{f}_{:j}$, $j=1,\dots,r$ if and only if they can be uniquely identified from $\m{z}$. \label{prop:recoveqident}
\end{prop}
\begin{proof} We first show the `if' part using a proof by contradiction. To do so, suppose that the atomic decomposition of the optimizer of \eqref{eq:problem} cannot recover the frequencies $\m{f}_{:j}$, $j=1,\dots,r$. This means that there exist $\m{f}'_{:j}$, $j=1,\dots,r'\leq r$ and $\m{y}'\in\m{\cC}$ satisfying that $\m{y}'=\sum_{j=1}^{r'} c'_j \m{a}\sbra{\m{f}'_{:j}}$. It follows that $\m{z}=\m{L}\m{y}' = \sum_{j=1}^r c'_j \m{L}\m{a}\sbra{\m{f}'_{:j}}$. Hence, the frequencies cannot be uniquely identified from $\m{z}$, which contradicts the condition of Proposition \ref{prop:recoveqident}.

Using similar arguments we can show that, if the frequencies cannot be uniquely identified from $\m{z}$, then they cannot be recovered from the optimizer of \eqref{eq:problem}. This means that the `only if' part also holds true.
\end{proof}

Proposition \ref{prop:recoveqident} establishes a link between the performance of frequency recovery using the atomic $\ell_0$ norm minimization and the parameter identifiability problem that has been well studied in the full data case where $\m{L}$ is an identity matrix, see \cite{sidiropoulos2001generalizing,jiang2001almost,liu2002almost, liu2006eigenvector}. It is well known that identifiability is a prerequisite for recovery. As a result, the atomic $\ell_0$ norm minimization provides the strongest theoretical guarantee possible.

\subsection{Precise Formulation of the Atomic $\ell_0$ Norm}

For the purpose of computation, a finite-dimensional formulation of the atomic $\ell_0$ norm is required. To do that, we assume that $\norm{\m{y}}_{\cA,0}=r<\overline{r}$, where $\overline{r}$ is known. In fact, we can always let $\overline r= N+1$ since $\m{y}$ can always be written as a linear combination of $N$ sinusoids with different frequencies; of course a tighter bound leads to lower computational complexity (see below).

For any $\m{f}_{:j}\in\bT^d$ and $n'_l\geq n_l$, $l=1,\dots,d$, note that $\m{a}\sbra{\m{f}_{:j}}=\m{a}_{\m{n}}\sbra{\m{f}_{:j}}$ is a subvector of $\m{a}'\sbra{\m{f}_{:j}} \coloneqq \m{a}_{\m{n}'}\sbra{\m{f}_{:j}}$. Let $\m{\Omega}$ be the index set which is such that $\m{a}\sbra{\m{f}_{:j}} = \m{a}'_{\m{\Omega}}\sbra{\m{f}_{:j}}$. Similarly, $\m{T} = \m{T}_{\m{n}}$ is a submatrix of $\m{T}' \coloneqq \m{T}_{\m{n}'} $. We have the following result.

\begin{thm} Assume that $\norm{\m{y}}_{\cA,0} =r < \overline{r}$. Let $n'_l\geq \max\sbra{n_l, \;\overline {r}}$, $l=1,\dots, d$.  Then, $\norm{\m{y}}_{\cA,0}$ equals the optimal value of the optimization problem
\equ{\begin{split}
&\min_{t,\m{T}',\m{y}'} \rank\sbra{ \m{T}'},\\
&\st \begin{bmatrix}t & \m{y}'^H \\ \m{y}' & \m{T}' \end{bmatrix} \geq \m{0},\; \m{y}'_{\m{\Omega}}=\m{y}, \end{split} \label{eq:rankmin_atom0norm}}
where the objective function $\rank\sbra{ \m{T}'}$ can be replaced by $\rank\sbra{\begin{bmatrix}t & \m{y}'^H \\ \m{y}' & \m{T}' \end{bmatrix}}$. \label{thm:atomic0norm}
\end{thm}

\begin{proof} Using the fact that $\norm{\m{y}}_{\cA,0} =r$, we have that $\m{y}$ admits an order-$r$ atomic decomposition as $\m{y}=\sum_{j=1}^r c_j\m{a}\sbra{\m{f}_{:j}}$.   Then, we can construct a feasible solution as
\equ{\begin{split}
&\sbra{t,\, \m{T}',\,\m{y}'} \\
&=\sbra{r,\, \sum_{j=1}^r \abs{c_j}^2\m{a}'\sbra{\m{f}_{:j}}\m{a}'^H\sbra{\m{f}_{:j}}, \, \sum_{j=1}^r c_j\m{a}'\sbra{\m{f}_{:j}}} \end{split}}
since $\m{y}'_{\m{\Omega}}= \sum_{j=1}^r c_j\m{a}'_{\m{\Omega}}\sbra{\m{f}_{:j}} = \sum_{j=1}^r c_j\m{a}\sbra{\m{f}_{:j}} = \m{y}$ and
\equ{\begin{split}
&\begin{bmatrix}t & \m{y}'^H \\ \m{y}' & \m{T}' \end{bmatrix} \\
&= \sum_{j=1}^r \begin{bmatrix}1 & c_j\m{a}'^H\sbra{\m{f}_{:j}} \\ c_j\m{a}'\sbra{\m{f}_{:j}} & \abs{c_j}^2\m{a}'\sbra{\m{f}_{:j}}\m{a}'^H\sbra{\m{f}_{:j}} \end{bmatrix} \\
&= \sum_{j=1}^r \begin{bmatrix}1 \\ c_j\m{a}'\sbra{\m{f}_{:j}} \end{bmatrix} \begin{bmatrix}1 \\ c_j\m{a}'\sbra{\m{f}_{:j}} \end{bmatrix}^H \\
&\geq \m{0}. \end{split} }
It follows that $r^*\leq \rank\sbra{\m{T}'} =r =\norm{\m{y}}_{\cA,0}$, where $r^*$ denotes the optimal solution of the problem in \eqref{eq:rankmin_atom0norm}.
On the other hand, when \eqref{eq:rankmin_atom0norm} achieves the optimal value $r^*$ at the optimizer $\sbra{t^*, \m{T}'^*,\m{y}'^*}$, we have that $\rank\sbra{\m{T}'^*}=r^*\leq r<\min_l n'_l$. It follows by Theorem \ref{thm:vanderdec} that $\m{T}'^*$ admits a unique Vandermonde decomposition as $\m{T}'^* = \sum_{j=1}^{r^*} \abs{c_j^*}^2\m{a}'\sbra{\m{f}^*_{:j}}\m{a}'^H\sbra{\m{f}^*_{:j}}$. Therefore, there exists $\m{s}^*$ such that $\m{y}'^* = \sum_{j=1}^{r^*} s_j^*\m{a}'\sbra{\m{f}^*_{:j}}$ since $\m{y}'^*$ lies in the range space of $\m{T}'^*$. It follows that $\m{y}=\m{y}'^*_{\m{\Omega}} = \sum_{j=1}^{r^*} s_j^*\m{a}\sbra{\m{f}^*_{:j}}$ and hence $\norm{\m{y}}_{\cA,0}\leq r^*$. So we conclude that $\norm{\m{y}}_{\cA,0}= r^*$.

When the objective function $\rank\sbra{ \m{T}'}$ is replaced by $\rank\sbra{\begin{bmatrix}t & \m{y}'^H \\ \m{y}' & \m{T}' \end{bmatrix}}$, the stated result follows by a similar argument.   In fact, in the first part of the proof, using the same constructed feasible solution we have that $r^* \leq \rank\sbra{\begin{bmatrix}t & \m{y}'^H \\ \m{y}' & \m{T}' \end{bmatrix}} = r =\norm{\m{y}}_{\cA,0}$. Then, in the second part of the proof, at the optimizer $\sbra{t^*, \m{T}'^*,\m{y}'^*}$, we have that $\rank\sbra{\m{T}'^*} \leq \rank\sbra{\begin{bmatrix}t^* & \m{y}'^{*H} \\ \m{y}'^* & \m{T}'^* \end{bmatrix}} =r^*\leq r<\min_l n'_l$. The same arguments can then be invoked to show that $\norm{\m{y}}_{\cA,0}\leq r^*$.
\end{proof}

By \eqref{eq:rankmin_atom0norm}, the problem in \eqref{eq:problem} can be written as the following rank minimization problem:
\equ{\begin{split}
\min_{t,\m{T}',\m{y}'} \rank\sbra{ \m{T}'}, \st \begin{bmatrix}t & \m{y}'^H \\ \m{y}' & \m{T}' \end{bmatrix} \geq \m{0},\; \m{y}'_{\m{\Omega}} \in \m{\cC}. \end{split} \label{eq:rankmin_problem}}
Once \eqref{eq:rankmin_problem} is solved, the estimate of $\m{y}$ is obtained as $\m{y}'_{\m{\Omega}}$, the frequencies can be retrieved from the Vandermonde decomposition of $\m{T}'$ which can be computed using the MaPP algorithm proposed in Section \ref{sec:MaPP}, and an atomic decomposition of $\m{y}$ follows naturally.

\begin{rem}
In practice, due to the possible absence of a tight upper bound $\overline r$ and/or computational considerations, we may also be interested in the following problem:
\equ{\min_{t,\m{T}} \rank\sbra{ \m{T}},\st \begin{bmatrix}t & \m{y}^H \\ \m{y} & \m{T} \end{bmatrix} \geq \m{0},  \label{eq:rankmin_atom0norm1}}
which corresponds to setting $n'_l=n_l$, $l=1,\dots,d$ in \eqref{eq:rankmin_atom0norm}.
Let $\sbra{t^*,\m{T}^*}$ be the optimal solution of \eqref{eq:rankmin_atom0norm1} with $r^*=\rank\sbra{\m{T}^*}$. Then, by the proof of Theorem \ref{thm:atomic0norm}, \eqref{eq:rankmin_atom0norm1} precisely characterizes $\norm{\m{y}}_{\cA,0}$ if $\m{T}^*$ admits a Vandermonde decomposition of order $r^*$, which is guaranteed by Theorem \ref{thm:vanderdec} if $r^*< \min_l n_l$. Otherwise, the existence of the Vandermonde decomposition of $\m{T}^*$ can be checked using the MaPP proposed in Section \ref{Sec:MaPP}. This provides a checking mechanism for determining whether the dimension-reduced problem in \eqref{eq:rankmin_atom0norm1} achieves $\norm{\m{y}}_{\cA,0}$. At the same time, it also provides an approach to frequency retrieval from the solution of \eqref{eq:rankmin_atom0norm1}. \label{rem:atom0norm_smalldim}
\end{rem}

\subsection{Solution via Convex Relaxation} \label{sec:CVX}
We have shown in Theorem \ref{thm:atomic0norm} that the atomic $\ell_0$ norm minimization is a rank minimization problem. For this type of problem the nuclear norm relaxation has been proven to be a practical and powerful tool \cite{recht2007guaranteed,tang2012compressed}. So we consider the following nuclear norm/trace minimization problem [by relaxing $\rank\sbra{\begin{bmatrix}t & \m{y}'^H \\ \m{y}' & \m{T}' \end{bmatrix}}$ to $\tr\sbra{\begin{bmatrix}t & \m{y}'^H \\ \m{y}' & \m{T}' \end{bmatrix}}=t+ \tr\sbra{\m{T}'}$]:
\equ{\begin{split}
&\min_{t,\m{T}',\m{y}'} t+ \tr\sbra{\m{T}'},\\
&\st \begin{bmatrix}t & \m{y}'^H \\ \m{y}' & \m{T}' \end{bmatrix} \geq \m{0},\; \m{y}'_{\m{\Omega}}=\m{y}. \end{split} \label{eq:tracemin}}
It turns out that \eqref{eq:tracemin} with \emph{appropriate choices} of $n'_l$, $l=1,\dots,d$ is nothing but (the dual problem of) the SDP formulation (up to a factor equal to $\frac{1}{2}$) of the atomic norm \cite{xu2014precise} defined as
\equ{\norm{\m{y}}_{\cA} \coloneqq \inf_{c_j\in\bC,\, \m{f}_{:j}\in\bT^d}\lbra{\sum_j \abs{c_j}:\; \m{y} = \sum_{j} c_j \m{a}\sbra{\m{f}_{:j}} }. }
The atomic norm was shown in \cite{chi2015compressive} to successfully recover the frequencies if they are \emph{sufficiently separated}. It is easy to show that the optimal objective value of \eqref{eq:tracemin} provides a lower bound for the atomic norm. However, the problem of choosing $n'_l$, $l=1,\dots,d$ such that \eqref{eq:tracemin} is guaranteed to characterize $\norm{\m{y}}_{\cA}$ is open. The paper \cite{xu2014precise} provides the only known checking mechanism, given the solution of \eqref{eq:tracemin}, but it involves a $d$D search over all frequencies at which the so-called dual polynomial achieves the maximum magnitude, and is hard to implement.

Using the Vandermonde decomposition of MLT matrices we can provide another checking mechanism as follows. It can be shown that \eqref{eq:tracemin} achieves $\norm{\m{y}}_{\cA}$ if the solution $\m{T}'^*$ admits a Vandermonde decomposition (see also \cite{chi2015compressive}). Therefore, similarly to what was said in Remark \ref{rem:atom0norm_smalldim}, a checking mechanism can be implemented by verifying whether $\rank\sbra{\m{T}'^*}<\min_l n'_l$ holds or otherwise finding a Vandermonde decomposition of $\m{T}'^*$ using MaPP. Compared to the dual polynomial method in \cite{xu2014precise}, this method requires less computations and is more practical.

\subsection{Solution via Reweighted Minimization}
Let
\equ{ \cM\sbra{\m{y}'} \coloneqq
\min_{t,\m{T}'} \rank\sbra{ \m{T}'}, \st \begin{bmatrix}t & \m{y}'^H \\ \m{y}' & \m{T}' \end{bmatrix} \geq \m{0}  \label{eq:rankmin_problem2}}
be a metric of $\m{y}'$. Then, \eqref{eq:rankmin_problem} can be rewritten as
\equ{\min_{\m{y}'} \cM\sbra{\m{y}'}, \st \m{y}'_{\m{\Omega}} \in \m{\cC}. \label{eq:rankmin_problem3}}
To approximately solve \eqref{eq:rankmin_problem3} [or \eqref{eq:rankmin_problem}, \eqref{eq:problem}], inspired by \cite{yang2014enhancing}, we propose a smooth surrogate for $\cM\sbra{\m{y}'}$:
\equ{\begin{split}
\cM^{\epsilon}\sbra{\m{y}'}= &\min_{t,\m{T}'} t + \ln\abs{\m{T}'+\epsilon \m{I}},\\
&\st \begin{bmatrix}t & \m{y}'^H \\ \m{y}' & \m{T}' \end{bmatrix} \geq \m{0}, \end{split} \label{eq:newsparsemetric}}
where $\epsilon>0$ is a regularization parameter.   In \eqref{eq:newsparsemetric}, $\ln\abs{\m{T}'+\epsilon\m{I}}$ is a smooth surrogate for $\rank\sbra{\m{T}'}$ in \eqref{eq:rankmin_problem2} and the additional term $t$ is included in the objective function to control the magnitude of $\m{T}'$ and avoid a trivial solution. Similarly to \cite{yang2014enhancing} we have the following result.

\begin{prop} Let $\epsilon\rightarrow 0$. Then, the following statements hold true:
\begin{enumerate}
\item If $\cM\sbra{\m{y}'}<N'\coloneqq\prod_{l=1}^d n_l'$, then
\equ{\cM^{\epsilon}\sbra{\m{y}}\sim \sbra{\cM\sbra{\m{y}'} - N'}\ln\frac{1}{\epsilon},}
i.e., $\lim_{\epsilon\rightarrow 0} \frac{\cM^{\epsilon}\sbra{\m{y}}} {\sbra{\cM\sbra{\m{y}'} - N'}\ln\frac{1}{\epsilon}} = 1$; otherwise, $\cM^{\epsilon}\sbra{\m{y}'}$ is a constant;
\item Let $\m{T}_{\epsilon}'^*$ be the optimizer of $\m{T}'$ in \eqref{eq:newsparsemetric}. Then, the smallest $N'-\cM\sbra{\m{y}'}$ eigenvalues of $\m{T}_{\epsilon}'^*$ are either zero or approach zero at least as fast as $\epsilon$, and any cluster point of $\m{T}_{\epsilon}'^*$ at $\epsilon=0$ has rank equal to $\cM\sbra{\m{y}'}$.
\end{enumerate} \label{prop:Mepsilony}
\end{prop}
\begin{proof} See Appendix \ref{append:Mepsilony}.
\end{proof}

Next, we consider the minimization of $\cM^{\epsilon}\sbra{\m{y}'}$. Therefore, rather than directly solving \eqref{eq:rankmin_problem}, we replace the objective function in \eqref{eq:rankmin_problem} by that in \eqref{eq:newsparsemetric}. Following the developments in \cite{yang2014enhancing}, a locally convergent iterative algorithm can be derived in which the $j$th iterate of $\m{T}'$, denoted by $\m{T}'_{j}$, is obtained as the optimizer of the following \emph{weighted trace minimization} problem:
\equ{\begin{split}
&\min_{t,\m{T}',\m{y}'} t + \tr\mbra{\sbra{\m{T}'_{j-1}+\epsilon_{j-1} \m{I}}^{-1}\m{T}'},\\
&\st \text{the constraints in \eqref{eq:rankmin_problem}}. \end{split} }
Here $\lbra{\epsilon_j>0:\; j\geq1}$ is a monotonically decreasing sequence. The resulting algorithm is designated as reweighted trace minimization (RWTM). Let $\m{T}'_{0}=\m{0}$ and $\epsilon_0=1$. Then the first iteration is exactly the convex relaxation method introduced in Section \ref{sec:CVX} (see \eqref{eq:tracemin}). The iterative reweighted process has the potential of enhancing sparsity and resolution (see the 1D case results in \cite{yang2014enhancing}). Note that a practical implementation of RWTM will trade off performance for computation time by keeping the number of iterations small.

%

\section{Numerical Simulations} \label{sec:simulation}

\begin{figure}
\centering
  \includegraphics[width=3.5in]{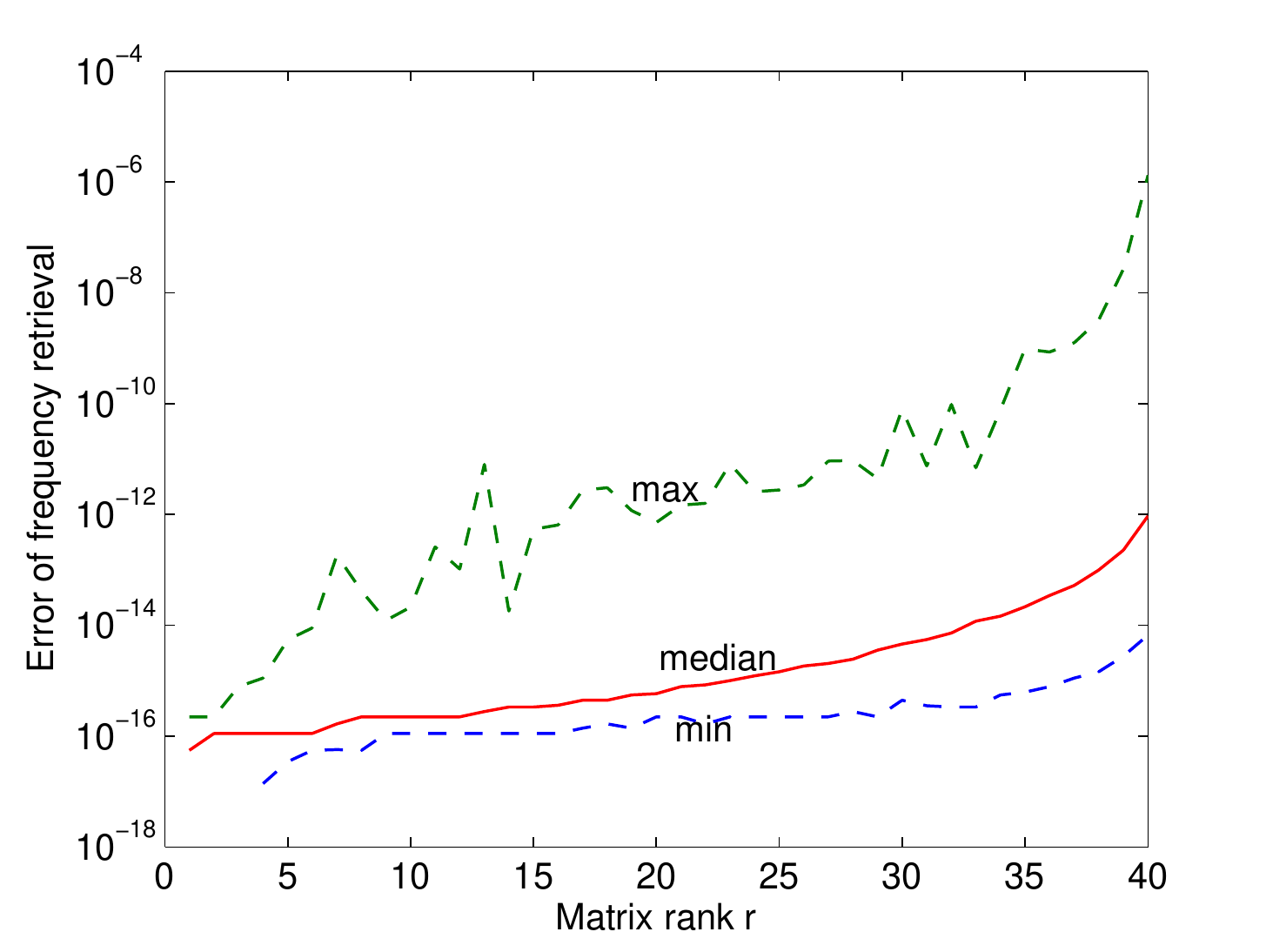}
\centering
\caption{The maximum, median and minimum errors of frequency retrieval in the Vandermonde decomposition result using MaPP. 1000 Monte Carlo runs are carried out for each $r$.} \label{Fig:VanDec}
\end{figure}

\subsection{Vandermonde Decomposition}

We numerically study the performance of the proposed MaPP algorithm for finding the Vandermonde decomposition. In particular, we consider a 2D case with $\m{n}=\sbra{6,8}$ and let $r=1,\dots,N=48$. In each problem instance, $2r$ frequencies are uniformly generated at random in $\bT$ and from them we form $r$ 2D frequencies $\m{f}_{:j}$, $j=1,\dots,r$. The power parameters $p_j$, $j=1,\dots,r$ are generated as $w^2+0.5$, where $w$ has a standard normal distribution. Then the 2LT matrix $\m{T}$ is obtained as in \eqref{eq:vanderdec}. After that, MaPP is used to find a Vandermonde decomposition of $\m{T}$ of order $r$. The error of frequency retrieval is measured as the maximum absolute error of the $2r$ frequencies. For each $r$, 1000 Monte Carlo runs are carried out. Note that MaPP only works when $r\leq N-\frac{N}{\text{min}\sbra{n_1,n_2}}=40$ since otherwise the matrix pencil will have eigenvalues equal to infinity. The simulation result for $r\leq 40$ is presented in Fig. \ref{Fig:VanDec}. Consistently with Theorem \ref{thm:algsucc} and Proposition \ref{prop:asshold}, it can be seen that for $r\leq 40$ MaPP can always retrieve the frequencies and find the Vandermonde decomposition within numerical precision. We observed that a relatively large numerical error occurs in the presence of closely spaced frequencies. Moreover, the numerical errors propagate quickly as $r$ gets close to 40.

\begin{figure*}
\centering
  \subfigure[]{
    \label{Fig:illus_RWTM}
    \includegraphics[width=3.5in]{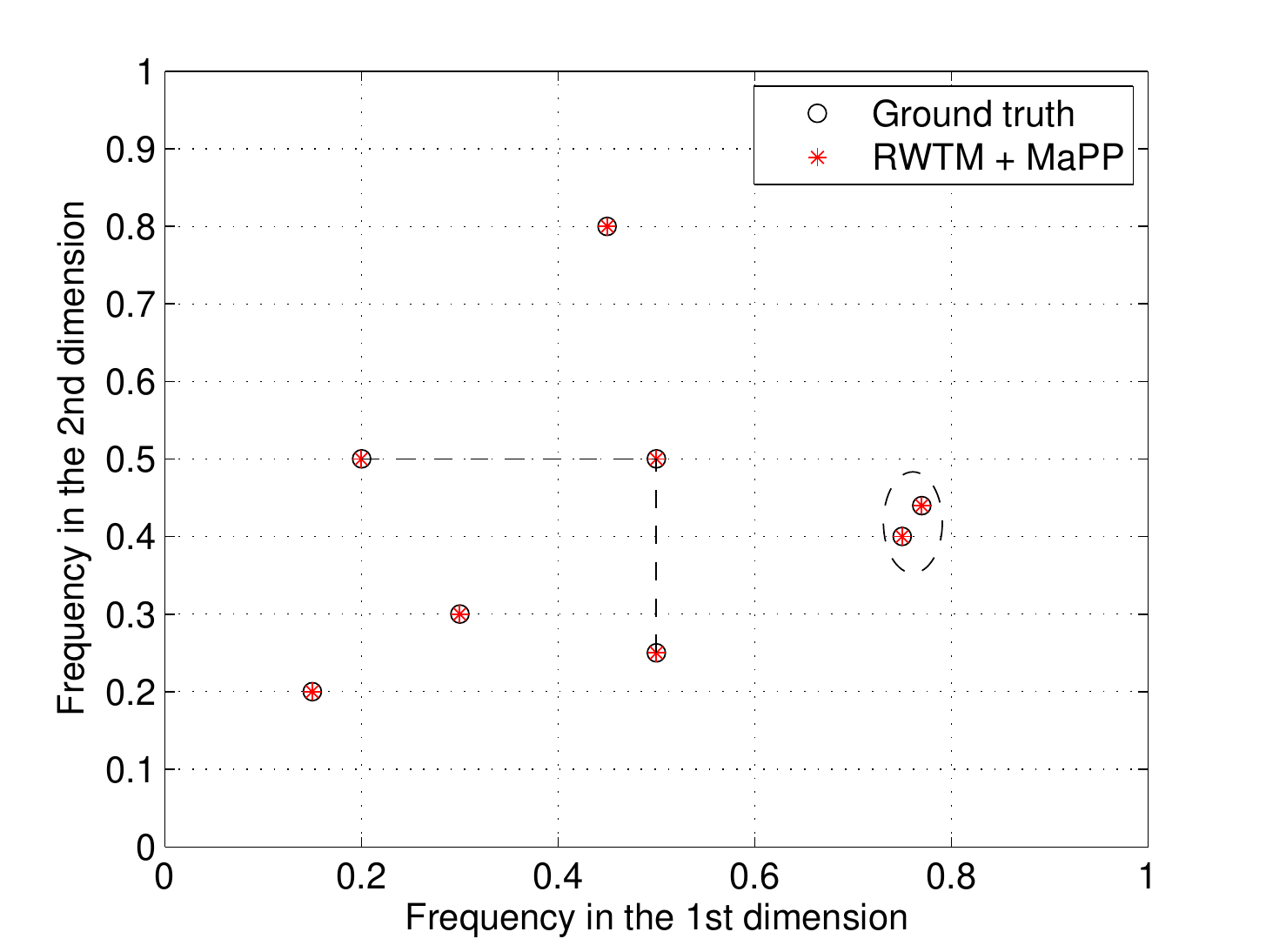}}%
  \subfigure[]{
    \label{Fig:illus_eigenvalue}
    \includegraphics[width=3.5in]{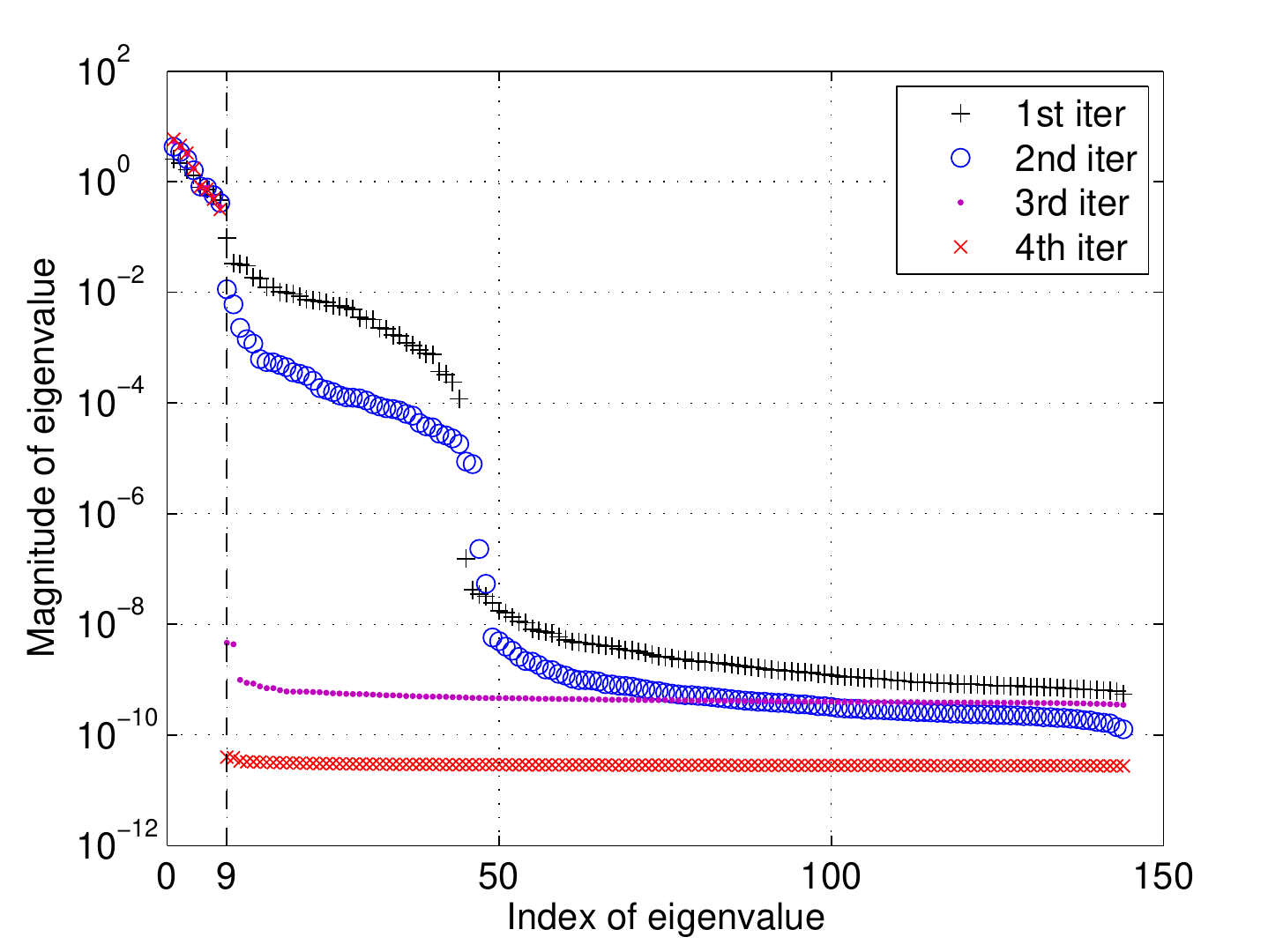}}
  \subfigure[]{
    \label{Fig:illus_contour_CVX}
    \includegraphics[width=3.5in]{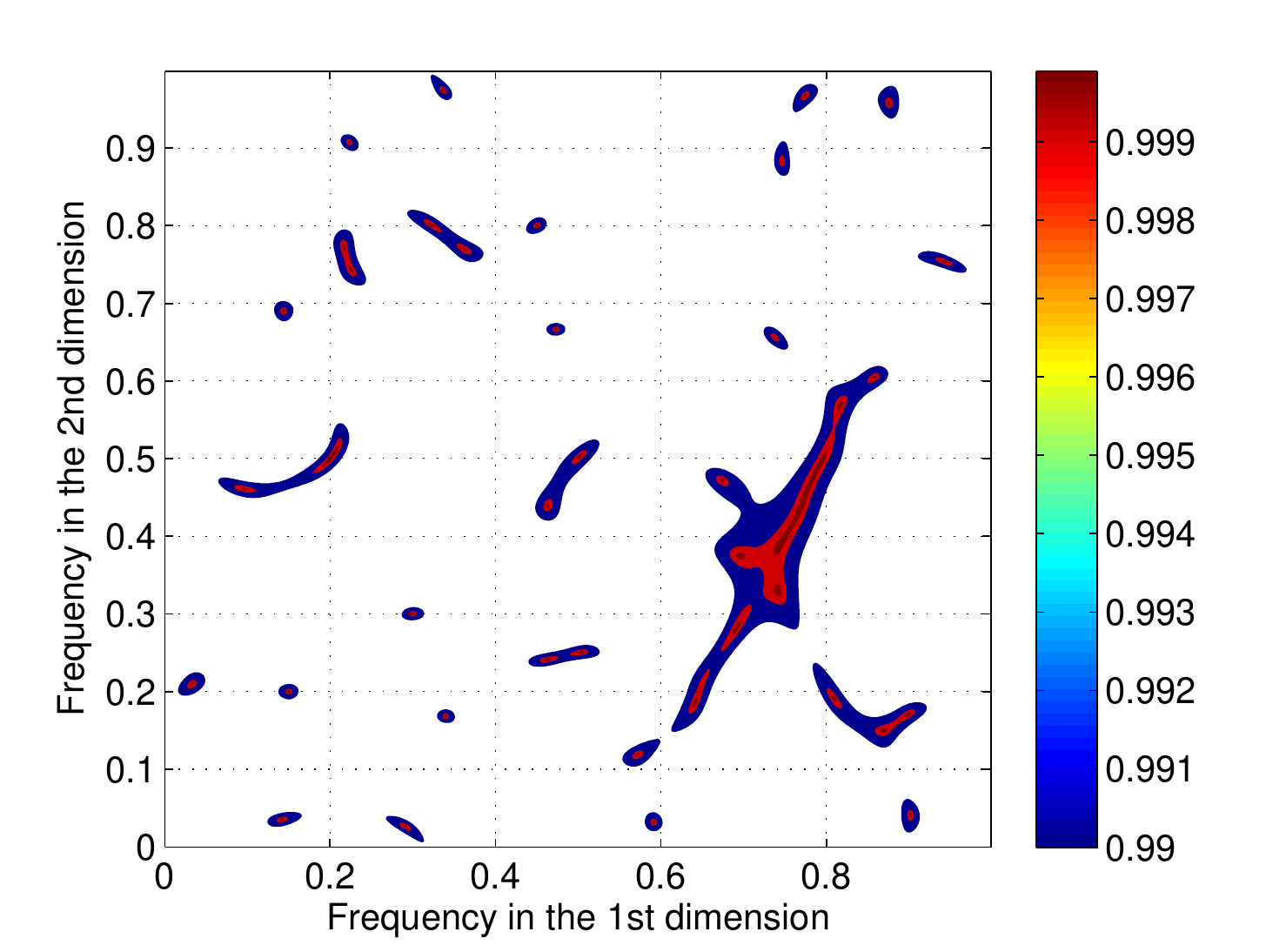}}%
  \subfigure[]{
    \label{Fig:illus_support_CVX}
    \includegraphics[width=3.5in]{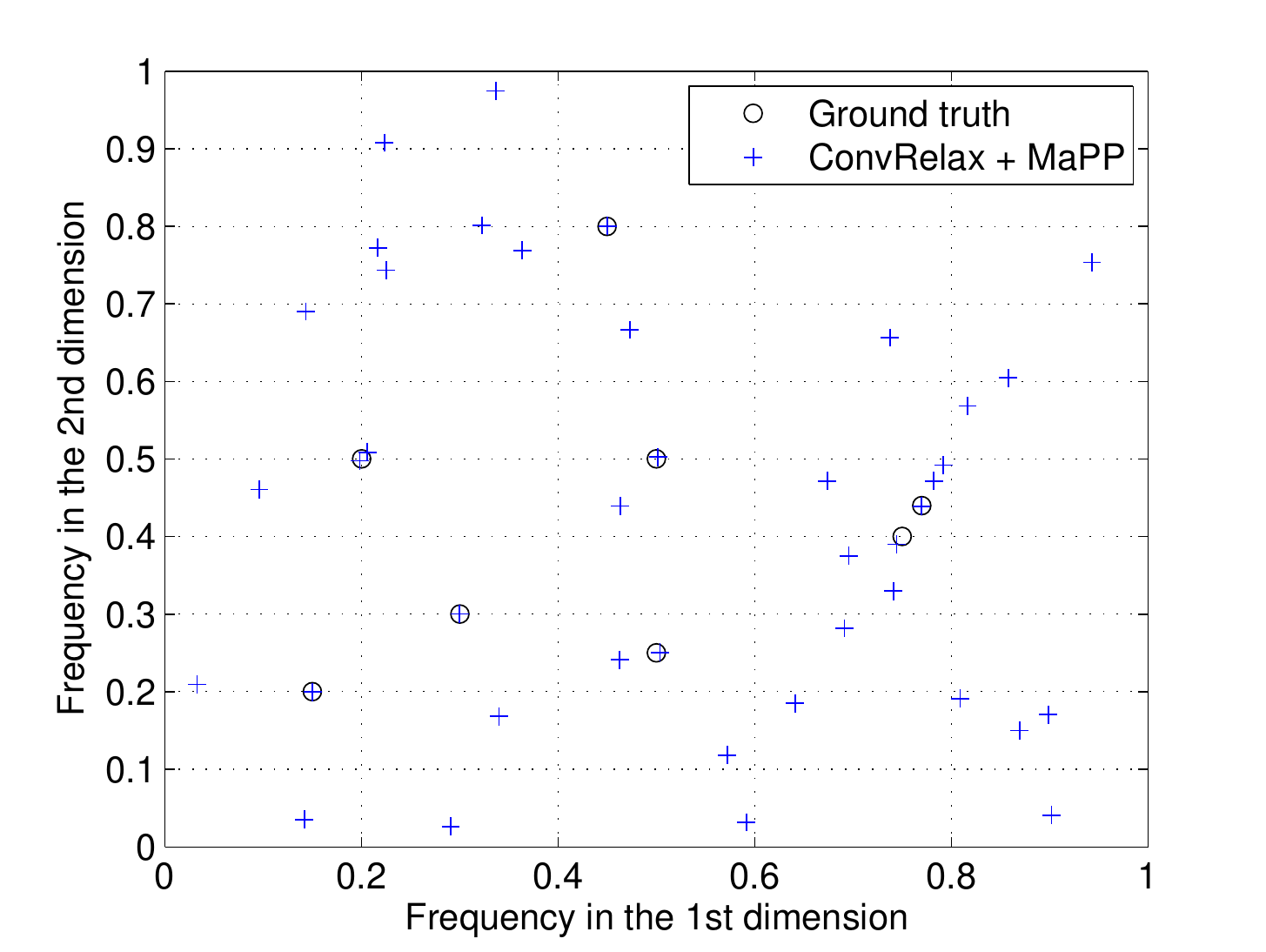}}
\centering
\caption{Super-resolution results in an illustrative example. (a) Ground truth and the frequencies recovered by RWTM + MaPP. (b) Eigenvalues of the iterates of $\m{T}'$. (c) Contour plot of the magnitude of the dual polynomial of ConvRelax. (d) The frequencies recovered by ConvRelax + MaPP (at the first iteration).} \label{Fig:simulation}
\end{figure*}

\subsection{Super-Resolution}
RWTM is implemented in Matlab and the involved SDPs are solved using SDPT3 \cite{toh1999sdpt3}. We set $\m{T}'_0=\m{0}$ and $\epsilon_0=1$, and so the first iteration of RWTM coincides with the convex relaxation method in Section \ref{sec:CVX}. We set $\epsilon_1$ to be equal to $0.1$ times the largest eigenvalue of $\m{T}'_1$ and let
\equ{\epsilon_j = \left\{ \begin{array}{ll}\frac{1}{2}\epsilon_{j-1}, & j=2,\dots,8; \\ \epsilon_8, & j>8.\end{array}\right.}
RWTM is terminated if the $\ell_2$ norm of the relative change of $\m{y}'$ at two consecutive iterations is lower than $10^{-6}$ or the maximum number of iterations, set to 20, is reached. Given the $\m{T}'$ obtained with RWTM, MaPP is used to retrieve the frequency estimates by finding its Vandermonde decomposition.

We first present an illustrative example to demonstrate the effectiveness of the proposed RWTM and frequency retrieval methods. The true values of eight ($r=8$) 2D frequencies are plotted in Fig. \ref{Fig:illus_RWTM}. Two of them share a common frequency value in the first dimension; two share a common value in the second dimension; and another two are closely located with a Euclidean distance of about $0.045$ (as indicated by the black dashed lines or circle). A number of 50 randomly located noiseless measurements are collected from $N=10\times 10$ uniform samples, with the complex amplitudes of the 2D sinusoids being randomly generated from a standard complex normal distribution. Note that the two closely located frequencies are separated by only $\frac{0.45}{n_1}$ which is much smaller than the resolution limit condition in \cite{candes2013towards,chi2015compressive} for the atomic norm method. Assume we know that $r< \overline r = 12$ and so let $n_1'=n_2'=12$. We use RWTM and MaPP to estimate the sinusoidal signal and the frequencies.
%

The RWTM algorithm ends in four iterations. The simulation results are presented in Fig. \ref{Fig:simulation}. It can be seen from Fig. \ref{Fig:illus_eigenvalue} that RWTM gradually reduces the rank of $\m{T}'$ and finally produces a solution of rank 8. From this low-rank solution, the proposed MaPP algorithm successfully retrieves the true frequencies as shown in Fig. \ref{Fig:illus_RWTM}. The estimation errors of the frequencies (in $\ell_2$ norm) are on the order of $10^{-11}$. It takes between 9.5s and 14.8s to run one iteration of RWTM.

It is worth noting that the convex relaxation method, designated by ConvRelax, at the first iteration does not produce a sufficiently low-rank $\m{T}'$. Consequently, it cannot correctly recover the frequencies and the signal. Fig. \ref{Fig:illus_contour_CVX} plots the magnitude of the dual polynomial whose peaks of magnitude 1 are used to identify the frequency poles. Note that in this example many frequency poles are present and a continuous band with magnitude greater than $0.9999$ is present around the two closely located true frequencies. It is a challenging task to accurately locate all these peaks using a 2D search, and therefore it is difficult for the checking mechanism in \cite{xu2014precise} to determine whether ConvRelax realizes the atomic norm method.

Next, we turn to the new checking mechanism proposed in Section \ref{sec:CVX} based on finding a Vandermonde decomposition of $\m{T}'$. In this example, we count the eigenvalues of $\m{T}'$ above a threshold of $10^{-6}$, which gives the number 44 that is used as the matrix rank in MaPP. After that, the proposed MaPP algorithm is used to find a Vandermonde decomposition of order 44 with a relative error, measured by $\frac{\frobn{\m{T}' - \sum_{j=1}^{44} p_j \m{a}\sbra{\m{f}_{:j}} \m{a}^H\sbra{\m{f}_{:j}} }^2} {\frobn{\m{T}'}^2}$, on the order of $10^{-12}$. So we can conclude that ConvRelax indeed realizes the atomic norm method. The frequencies in the Vandermonde decomposition are presented in Fig. \ref{Fig:illus_support_CVX}, which well match the peaks of the dual polynomial shown in Fig. \ref{Fig:illus_contour_CVX}.

In the following simulation we study the sparse recovery capabilities of RWTM and ConvRelax in terms of sparsity-separation phase transition that was first introduced in \cite{yang2014enhancing}. We consider a 2D case with $\m{n}=\sbra{8,8}$. In each Monte Carlo run, a number of $32$ randomly located noiseless measurements are collected from the $N=64$ uniform samples, with the complex amplitudes of the 2D sinusoids being randomly generated from a standard complex normal distribution. The number of sinusoids $r$ that we consider ranges from 1 to 16. Any two 2D frequencies are separated (in $\ell_{\infty}$ norm following \cite{candes2013towards,chi2015compressive}) by at least $\Delta_f$ which takes on the values $0,\frac{0.1}{n_1},\dots,\frac{2}{n_1}$. To randomly generate a set of 2D frequencies, a new 2D frequency is randomly generated at one time and added to the set if the separation condition is satisfied and the process is repeated until $r$ frequencies are obtained. For each pair $\sbra{\Delta_f, r}$, 20 data instances are generated and the signal and its frequencies are estimated using RWTM and ConvRelax. For speed consideration we set $\m{n}'=\m{n}$ in both RWTM and ConvRelax. Successful recovery is declared if the relative root mean squared error (RMSE) of signal recovery is less than $10^{-6}$ and the error of frequency recovery (in $\ell_{\infty}$ norm) is less than $10^{-6}$ (our experience suggests that these two conditions are satisfied or violated jointly).

\begin{figure*}
\centering
  \subfigure[]{
    \label{Fig:succrate_CVX}
    \includegraphics[width=3.5in]{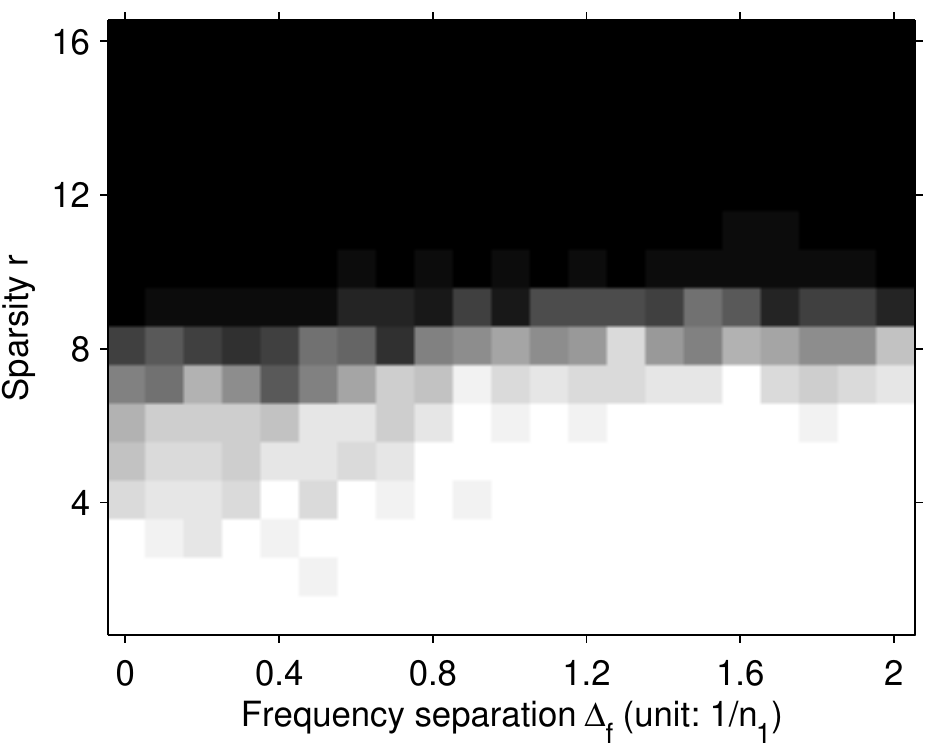}}%
  \subfigure[]{
    \label{Fig:succrate_RWTM}
    \includegraphics[width=3.5in]{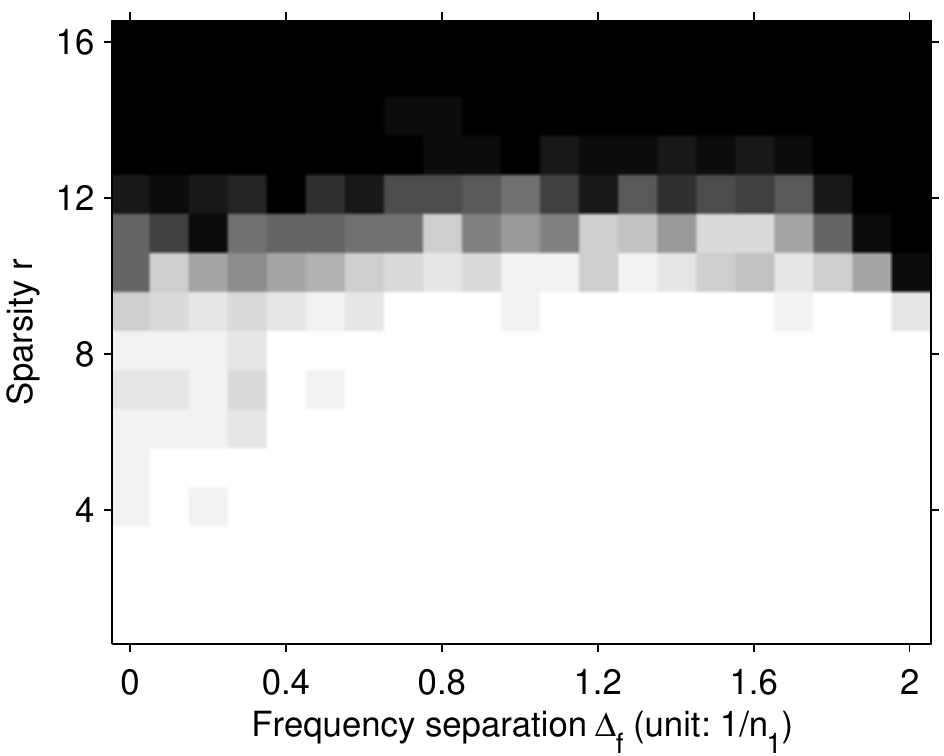}}
  \subfigure[]{
    \label{Fig:succrate_RWTMnotCVX}
    \includegraphics[width=3.5in]{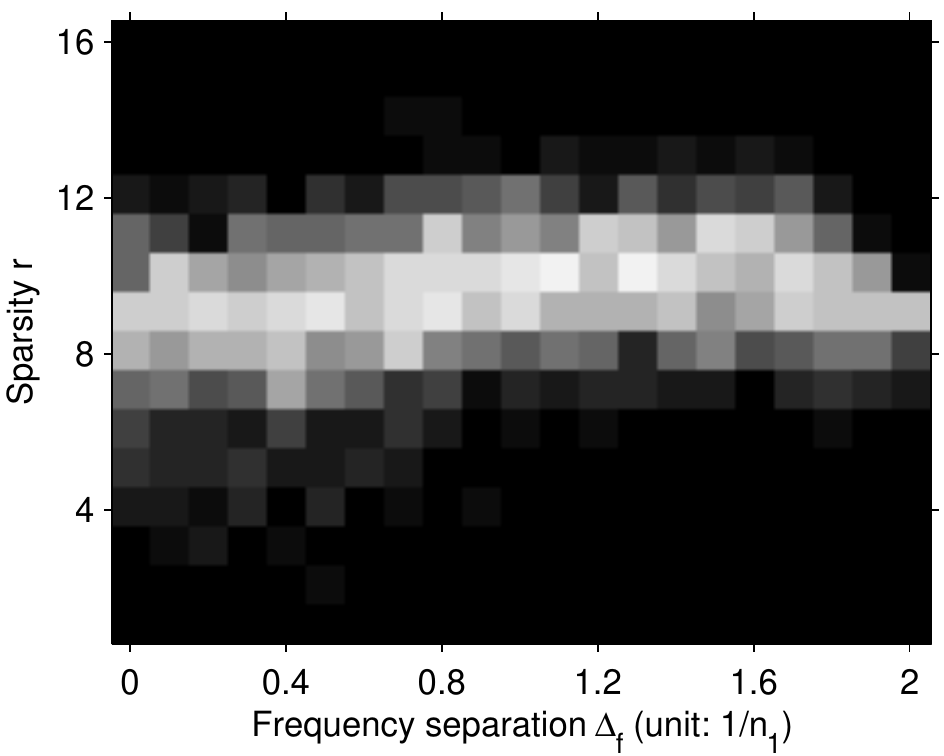}}%
  \subfigure[]{
    \label{Fig:succrate_genfreq}
    \includegraphics[width=3.5in]{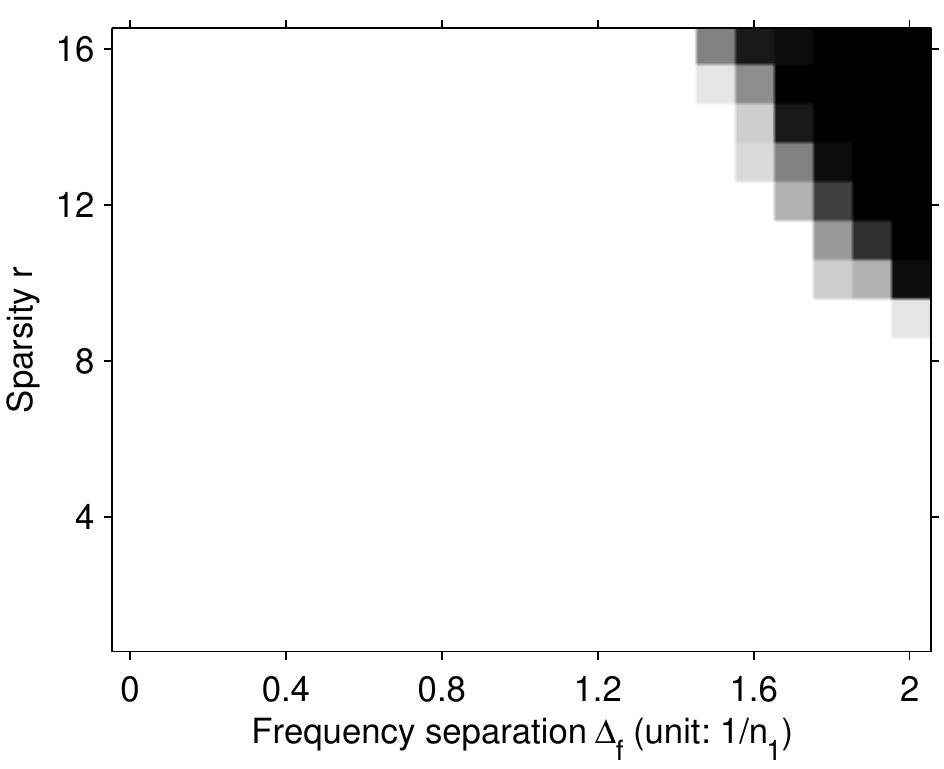}}
\centering
\caption{Super-resolution results. (a) Success rates of ConvRelax. (b) Success rates of RWTM. (c) The percentage of the total number of cases in which RWTM succeeds but ConvRelax fails. (d) Success rates of generating the set of 2D frequencies. Here white means complete success while black means complete failure.} \label{Fig:phasetrans}
\end{figure*}

The simulation results are presented in Fig. \ref{Fig:phasetrans}. In particular, Fig. \ref{Fig:succrate_CVX} and Fig. \ref{Fig:succrate_RWTM} plot, respectively, the success rates of ConvRelax and RWTM, where white means complete success and black means complete failure. Both methods can recover the signal and the frequencies in the regime of few sinusoids and large frequency separation, leading to phase transitions in the sparsity-separation plane. The phase transitions are not sharp, as in \cite{yang2014enhancing}, since the frequencies are separated by \emph{at least} $\Delta_f$ and well separated frequencies can still be obtained for small values of $\Delta_f$. RWTM clearly has an larger success region compared to ConvRelax. To illustrate this more clearly, we plot in Fig. \ref{Fig:succrate_RWTMnotCVX} the percentage of the total number of cases in which RWTM succeeds but ConvRelax fails. It can be seen that a significant number of the generated problems are of such a type, and that they are concentrated in the regime of median sparsity and/or small frequency separation. So, compared to ConvRelax, RWTM has improved sparsity recovery capability and enhanced resolution. Finally, it should be noted that not all problem instances are readily generated. To be specific, the frequency set is hard to generate in a reasonable amount of time in the regime of large $r$ and large $\Delta_f$ (see Fig. \ref{Fig:succrate_genfreq}). ConvRelax and RWTM have low success rates in this regime partly due to this reason.

\begin{figure*}
\centering
  \subfigure[]{
    \label{Fig:comp_MSE1}
    \includegraphics[width=3.5in]{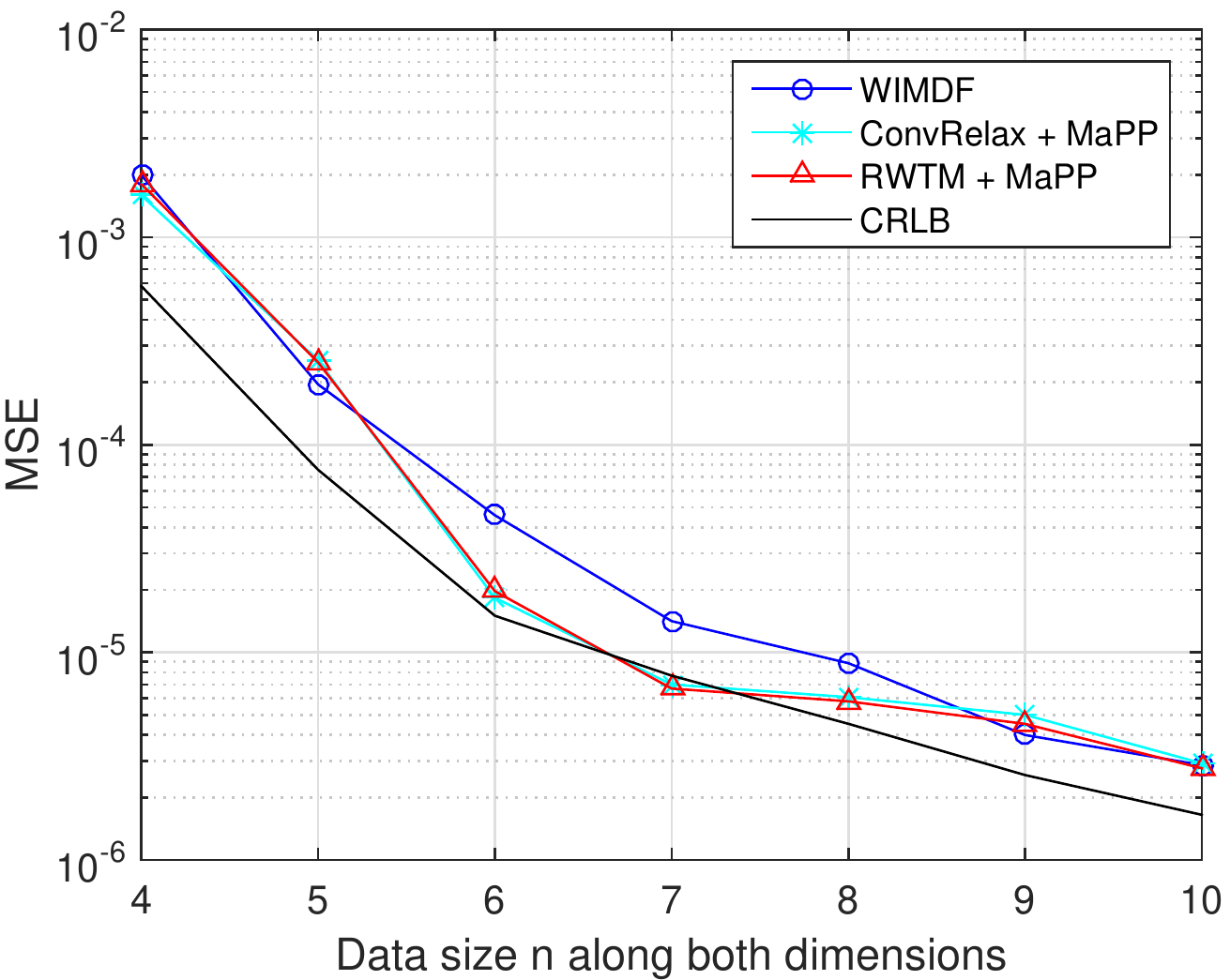}}%
  \subfigure[]{
    \label{Fig:comp_MSE_compressive}
    \includegraphics[width=3.5in]{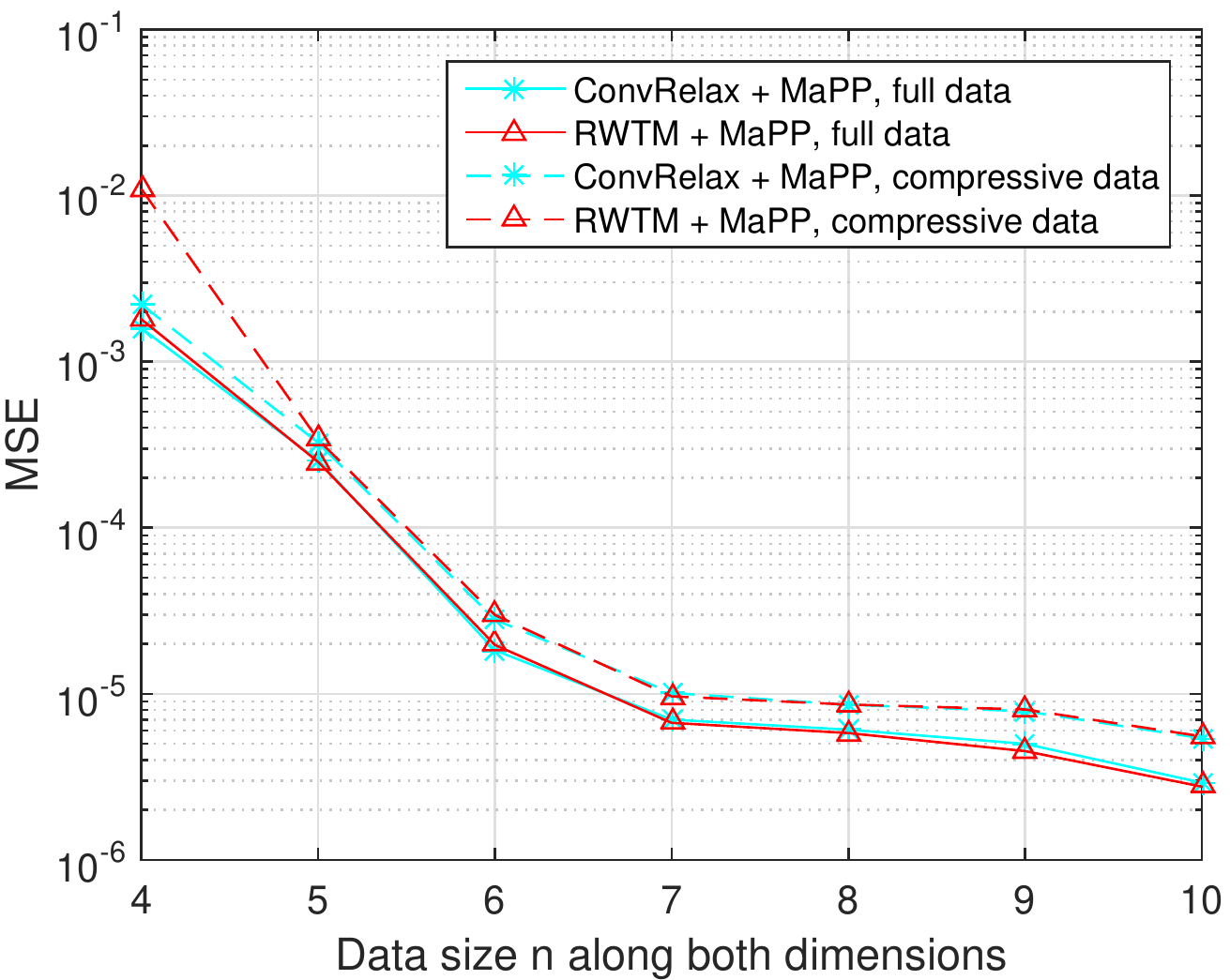}}
\centering
\caption{MSE results of 2D frequency estimation in the noisy case. (a) The full data case in which $n\times n$ uniform samples are acquired. (b) The compressive data case in which $80\%$ of the data samples are used to estimate the frequencies.} \label{Fig:comp_MSE}
\end{figure*}

  Finally, we consider the noisy case and compare the proposed super-resolution methods with a conventional subspace method. We select the weighted improved multidimensional folding (WIMDF) algorithm in \cite{liu2007multidimensional} as a benchmark. Note that the WIMDF algorithm (as in fact other subspace methods) does not work in the compressive data case and thus it is only considered in the full data case. In contrast to this, the proposed RWTM and ConvRelax methods are also considered in a compressive data case in which $80\%$ of the data samples (randomly selected) are used for frequency estimation. The proposed methods are implemented as in the noiseless case but with the feasible set $\m{\cC}$ in \eqref{eq:problem} re-defined as
\equ{\m{\cC} = \lbra{\m{y}:\; \twon{\m{z} - \m{L}\m{y}}^2\leq \eta^2}, \label{eq:Cwithnoise}}
where $\m{z}\in\bC^{M}$ denotes the vector of noisy measurements, $\m{L}\in\lbra{0,1}^{M\times N}$ is the selection matrix representing the sampling scheme, $M$ is the sample size, and $\eta^2$ is an upper bound on the noise energy. Note that $\m{L}$ is the identity matrix and $M=N$ in the full data case. We also note that RWTM will be terminated in five iterations.

In this simulation, we consider a 2D case with three ($r=3$) sinusoids that have frequencies $\sbra{0.25, 0.55}$, $\sbra{0.45, 0.55}$ and $\sbra{0.45, 0.35}$ and amplitudes $e^{i0.793\pi}$, $e^{i0.385\pi}$ and $e^{i0.076\pi}$, respectively. We let $n_1=n_2=n$ and consider different values of $n$ ranging from 4 to 10. We add white complex Gaussian noise to the $n\times n$ uniform samples and let the noise variance be $\sigma^2=\frac{0.1}{n}$ so that the signal to noise ratio (SNR) of the samples acquired following the data model in \eqref{eq:paramodel2} is approximately constant (in our simulation, the averaged SNR for each $n$ was between $14.6$ and $15.3$ dB). We set $\eta^2=\sbra{M+2\sqrt{M}}\sigma^2$ (i.e., mean + twice standard deviation) to upper bound the noise energy with high probability. This means that the noise variance is assumed to be known in the proposed methods while, somewhat similarly, WIMDF assumes that the number of frequencies $K$ is known. Since the proposed methods might produce spurious frequencies, the strongest three frequency components are used as the frequency estimates based on which the mean squared error (MSE) is computed (over 100 Monte Carlo runs). Regarding this aspect we note that RWTM + MaPP rarely overestimates the number of frequencies: this happened in only 4 out of 1400 Monte Carlo runs in our simulation (mainly due to the fact that the true noise energy can exceed $\eta^2$); on the other hand, ConvRelax + MaPP produced spurious frequencies with relatively weak powers in about $40\%$ of the runs.

\begin{figure*}
\centering
    \includegraphics[width=7in]{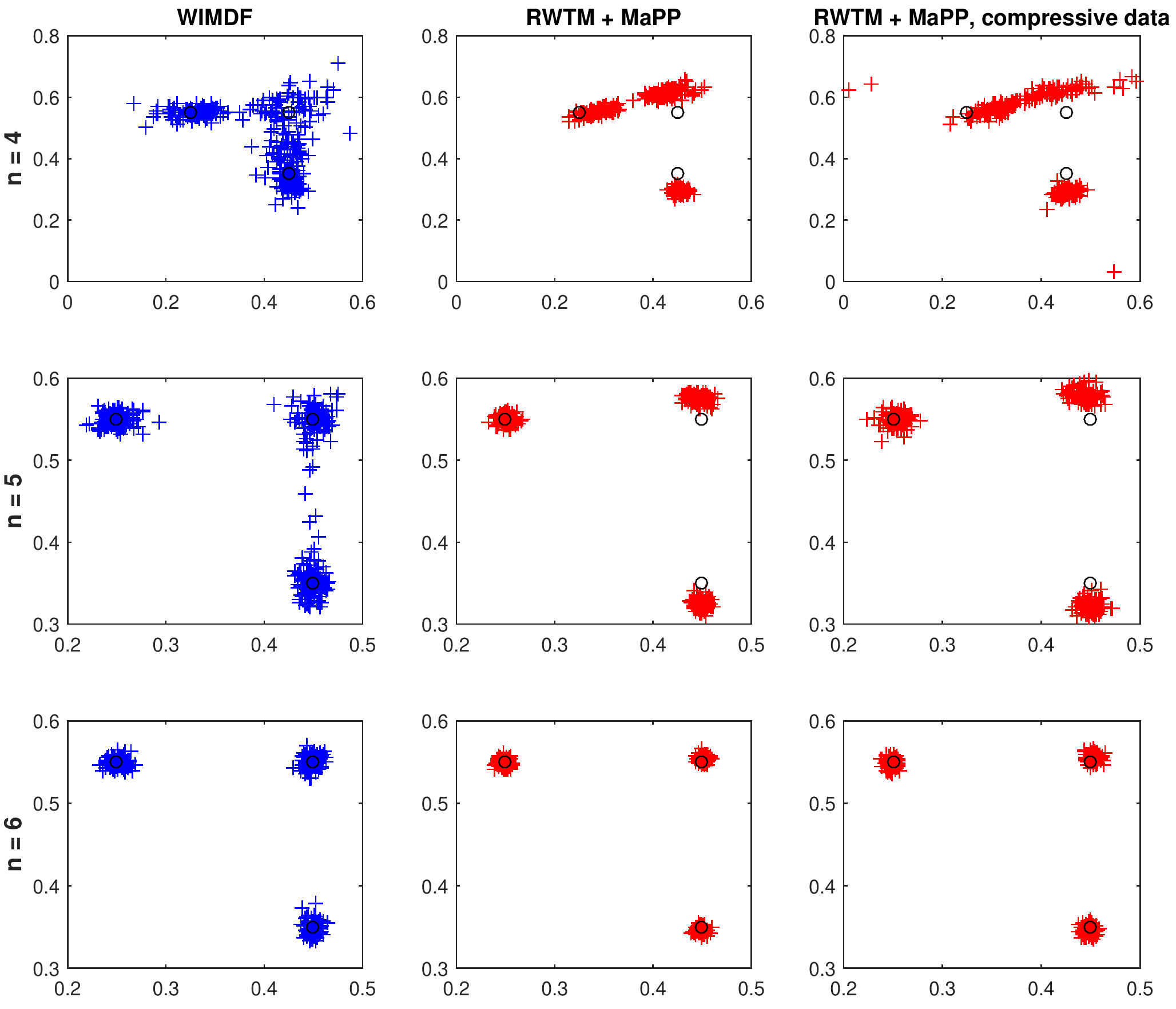}
\centering
\caption{2D frequency estimation results in the noisy case. The true frequencies are indicated using black circles while their estimates are shown as blue/red `+'. 1st row: $n=4$; 2nd row: $n=5$; 3rd row: $n=6$. 1st column: WIMDF with full data; 2nd column: RWTM + MaPP with full data; 3rd column: RWTM + MaPP with compressive ($80\%$) data.} \label{Fig:comp_freqest}
\end{figure*}

The simulation results on the MSE of frequency estimation are presented in Fig. \ref{Fig:comp_MSE}. We first note that the MSE curves of RWTM and ConvRelax almost coincide with each other. This implies that in this example ConvRelax + MaPP can accurately localize the true frequencies by selecting the strongest frequency components and the main advantage of RWTM + MaPP is to eliminate the spurious frequencies that ConvRelax + MaPP produces. In the left subfigure, the proposed methods are compared with WIMDF as well as the Cramer-Rao lower bound (CRLB), which is computed following \cite{liu2007multidimensional}, in the full data case. It can be seen that in all the scenarios considered RWTM + MaPP is either comparable with or better than WIMDF. For $n=6$ and $n=7$, compared to WIMDF, the MSE improvement of RWTM + MaPP is more than 3 dB. The right subfigure shows that for the proposed super-resolution methods $20\%$ data loss results in only a small degradation of 1.3 to 3 dB for RWTM + MaPP when $n\geq 5$. For $n=4$, the performance loss for RWTM + MaPP is larger because in some runs (14 out of 100) RWTM + MaPP can only detect two sinusoids. In contrast, WIMDF cannot work at all in the compressive data case.

We next show the frequency estimates of WIMDF and RWTM + MaPP in Fig. \ref{Fig:comp_freqest} (the results for ConvRelax + MaPP are omitted due to the presence of several spurious frequencies). For $n=4$, while WIMDF and RWTM + MaPP are comparable in terms of MSE, see Fig. \ref{Fig:comp_MSE}, it can be seen in Fig. \ref{Fig:comp_freqest} that the three frequencies can be more clearly separated by RWTM + MaPP (with a small bias though). The same behavior can be observed for $n=5$ as well. It can also be seen from this figure that for RWTM + MaPP the $20\%$ data loss causes only a small performance degradation for $n=5$ and $n=6$.

Since WIMDF is a subspace method and its main computations include a singular value decomposition and a small-scale optimization problem for parameter tuning, it is very fast in practice and needs about 0.1s on average for a single run. Because RWTM + MaPP adopts a more sophisticated optimization method, it requires 3 to 38s on average when $n$ increases from 4 to 10, while the computational time of ConvRelax + MaPP is about 1/5 of that required by RWTM + MaPP. The increased computational cost of RWTM + MaPP and ConvRelax + MaPP can be justified by the fact that, unlike WIMDF, they can be applied to the compressive data case. It is also worth noting that the proposed methods can be accelerated using faster solvers for SDP, e.g., the ADMM algorithm \cite{boyd2011distributed} (see \cite{bhaskar2013atomic, yang2015gridless,yang2014enhancing} for examples in the 1D case).

\section{Conclusion} \label{sec:conclusion}
In this paper, the Vandermonde decomposition of Toeplitz matrices was generalized from the 1D to the MD case under a rank condition. When this condition is not satisfied a numerical approach was also proposed for finding a possible decomposition. The new results were used to study the MD super-resolution problem and practical algorithms were proposed. Extensive numerical simulations were provided to validate our theoretical findings and demonstrate the effectiveness of the proposed super-resolution methods.

The result on Vandermonde decomposition presented in this paper is closely related to operator theory and structured linear algebra. Its application to these areas would be of interest. When the matrix rank is high, the question on existence of the Vandermonde decomposition is still open, which should also be studied in the future.

\section*{Acknowledgement}
We would like to thank Dr. Fredrik Andersson and Dr. Marcus Carlsson of Lund University, Sweden, for helpful discussions on the proof of Lemma \ref{lem:Tdiag}.

\appendix

\subsection{Proof of Lemma \ref{lem:Tdiag}} \label{append:Tdiag}
To prove Lemma \ref{lem:Tdiag} for $d=1$, we will use the following result.
\begin{lem} Consider a Hankel matrix $\m{H}\in\bC^{n\times n}$ defined as
\equ{\m{H} = \begin{bmatrix}h_1 & h_2 & \dots & h_n \\ h_2 & h_3 & \dots & h_{n+1} \\ \vdots & \vdots & \ddots & \vdots \\ h_n & h_{n+1} & \dots & h_{2n-1} \end{bmatrix}.}
If $\m{H}$ can be written as
\equ{\m{H} = \m{A}\sbra{\m{f}} \m{C} \m{A}^T\sbra{\m{f}}, \label{eq:HABAT}}
where $\m{C}\in\bC^{r\times r}$, $r<n$, and $f_j$, $j=1,\dots,r$ are distinct points in $\bT$, then $\m{C}$ must be a diagonal matrix. \label{lem:Hankel}
\end{lem}

\begin{proof} We make use of the Kronecker's theorem for Hankel matrices (see, e.g., \cite{ellis1992factorization}). Let $r'=\rank\sbra{\m{C}}\leq r$. Also let $\m{H}_{n-1}\in\bC^{\sbra{n-1}\times\sbra{n-1}}$ be the principal submatrix of $\m{H}$ obtained by removing the last row and column. It follows that
\equ{\m{H}_{n-1} = \m{A}_{n-1}\sbra{\m{f}}\m{C} \m{A}_{n-1}^T\sbra{\m{f}}.}
As both $\m{A}\sbra{\m{f}}$ and $\m{A}_{n-1}\sbra{\m{f}}$ have full column rank provided that $r<n$, it holds that $\rank\sbra{\m{H}} = \rank\sbra{\m{H}_{n-1}} = \rank\sbra{\m{C}} = r'$. By \cite[Theorem 3.1]{ellis1992factorization} $\m{H}$ can be factorized as
\equ{\m{H} = \widetilde{\m{A}}\widetilde{\m{C}}\widetilde{\m{A}}^T,}
where $\widetilde{\m{A}}\in\bC^{n\times r'}$ is a generalized Vandermonde matrix and $\widetilde{\m{C}}\in\bC^{r'\times r'}$ is an invertible block diagonal matrix. Interested readers are referred to \cite{ellis1992factorization} for the specific forms of $\widetilde{\m{A}}$ and $\widetilde{\m{C}}$. We will use the following facts: 1) any $n\times r'$ generalized Vandermonde matrix has full column rank if $n\geq r'$, and 2) $\widetilde{\m{C}}$ becomes a diagonal matrix if $\widetilde{\m{A}}$ is a Vandermonde matrix. From the equality
\equ{\widetilde{\m{A}}\widetilde{\m{C}}\widetilde{\m{A}}^T = \m{A} \m{C} \m{A}^T \label{eq:equality1}}
it follows that $\widetilde{\m{A}} = \m{A}\mbra{\m{C} \m{A}^T \sbra{\widetilde{\m{C}}\widetilde{\m{A}}^T}^{\dag}}$. This means that each column $\widetilde{\m{a}}_j$, $j=1,\dots,r'$ in $\widetilde{\m{A}}$ is a linear combination of the columns in $\m{A}$ and thus the $n\times (r+1)$ matrix $\mbra{\widetilde{\m{a}}_j, \m{A}}$ is rank-deficient. By the assumption that $n\geq r+1$ and the properties of generalized Vandermonde matrices mentioned above, it follows that $\widetilde{\m{a}}_j$ is a column in $\m{A}$ and thus $\widetilde{\m{A}}$ is a Vandermonde matrix formed by $r'$ columns in $\m{A}$. It also follows that $\widetilde{\m{C}}$ is a diagonal matrix. As a result, we conclude by \eqref{eq:equality1} that $\m{C}$ is a diagonal matrix whose diagonal consists of that in $\widetilde{\m{C}}$ and $r-r'$ zeros up to re-sorting.
\end{proof}

Now we can prove Lemma \ref{lem:Tdiag} in the case of $d=1$ for which \eqref{eq:TBCBH} becomes
\equ{\m{T} = \m{A}\sbra{\m{f}}\m{C}\m{A}^H\sbra{\m{f}}.}
Let $\m{H}$ and $\check{\m{A}}$ be the matrices obtained by sorting the rows of $\m{T}$ and $\m{A}$ in reverse order, respectively. It is obvious that $\m{H}$ is a Hankel matrix and $\m{H} = \check{\m{A}}\m{C} \m{A}^H$. Moreover, note that $\check{\m{A}} = \overline{\m{A}}\diag\sbra{e^{i2\pi\sbra{n-1}f_1}, \dots, e^{i2\pi\sbra{n-1}f_r}}$, where $\overline{\cdot}$ denotes the complex conjugate operator. It follows that
\equ{\overline{\m{H}} = \m{A}\diag\sbra{e^{-i2\pi\sbra{n-1}f_1}, \dots, e^{-i2\pi\sbra{n-1}f_r}} \overline{\m{C}} \m{A}^T }
which is still a Hankel matrix. By Lemma \ref{lem:Hankel} we conclude that $\diag\sbra{e^{-i2\pi\sbra{n-1}f_1}, \dots, e^{-i2\pi\sbra{n-1}f_r}} \overline{\m{C}}$ is a diagonal matrix and so is $\m{C}$.

Now suppose that Lemma \ref{lem:Tdiag} holds for $d=d_0-1$, $d_0\geq2$. By induction it suffices to show that Lemma \ref{lem:Tdiag} also holds for $d=d_0$. Let us view $\m{T}$ as an $n_1\times n_1$ block Toeplitz matrix. For the $\sbra{j+1,k+1}$th block, $0\leq j, k \leq n_1-1$ the following identity holds by \eqref{eq:TBCBH}:
\equ{\begin{split}
&\m{T}\sbra{j+1,k+1} \\
&= \m{A}_{\m{n}_{-1}}\sbra{\m{f}_{-1}} \m{Z}_1^j\m{C}\m{Z}_1^{-k}\m{A}_{\m{n}_{-1}}^H\sbra{\m{f}_{-1}}, \end{split}\label{eq:Tjkblk}}
where $\m{f}_{-1}$ denotes $\m{f}$ after removing the first row and $\m{Z}_1 = \diag\sbra{e^{i2\pi f_{11}},\dots, e^{i2\pi f_{1r}}}$.
Let $\lbra{\widetilde{\m{f}}_{-1,j}}_{j=1}^{r'_1}$, $r'_1\leq r$, denote the non-redundant collection of $\lbra{\m{f}_{-1,j}}_{j=1}^r$ and let $\m{\Gamma}\in\lbra{0,1}^{r'_1\times r}$ be the matrix which is such that $\m{f}_{-1} = \widetilde{\m{f}}_{-1} \m{\Gamma}$. Then, \eqref{eq:Tjkblk} becomes
\equ{\begin{split}
&\m{T}\sbra{j+1,k+1} \\
&= \m{A}_{\m{n}_{-1}}\sbra{\widetilde{\m{f}}_{-1}} \m{\Gamma}\m{Z}_1^j\m{C}\m{Z}_1^{-k}\m{\Gamma}^T \m{A}_{\m{n}_{-1}}^H\sbra{\widetilde{\m{f}}_{-1}}. \end{split}}
Note that $\m{T}\sbra{j+1,k+1}$, $0\leq j, k \leq n_1-1$ are all $\m{n}_{-1}$, $\sbra{d_0-1}$LT matrices and $\widetilde{\m{f}}_{-1,j}$, $j=1,\dots,r'_1$ are distinct points in $\bT^{d_0-1}$. By the assumption that Lemma \ref{lem:Tdiag} holds for $d=d_0-1$ we have that
\equ{\m{D}\sbra{j,k} \coloneqq \m{\Gamma}\m{Z}_1^j\m{C}\m{Z}_1^{-k}\m{\Gamma}^T, \quad 0\leq j, k \leq n_1-1 \label{eq:Djk}}
are all diagonal matrices.

Let $\widetilde{\m{C}}\sbra{j,k} = \m{Z}_1^j\m{C}\m{Z}_1^{-k}$. Its $\sbra{p,q}$th entry satisfies the equation:
\equ{\widetilde{C}_{pq}\sbra{j,k} = C_{pq}e^{i2\pi\sbra{jf_{1p}-kf_{1q}}}. \label{eq:Cpqjk}}
We next study some properties of $\m{\Gamma}$. Define $\m{S}_m=\lbra{j:\; \Gamma_{mj}=1}$, $m=1,\dots,r'_1$. According to the definition of $\m{\Gamma}$ it holds that
\equ{\m{S}_m=\lbra{j:\; \m{f}_{-1,j}=\widetilde{\m{f}}_{-1,m}}. \label{eq:Omegam}}
Therefore, $\m{S}_m$, $m=1,\dots,r'_1$ are disjoint subsets of $\lbra{1,\dots, r}$ with
$\bigcup_{m=1}^{r'_1} \m{S}_m = \lbra{1,\dots, r}$. Moreover, $f_{1p}$, $p\in\m{S}_m$ are distinct for any $m=1,\dots,r'_1$ since $\m{f}_{:j}$, $j=1,\dots,r$ are distinct.

Let $\m{\Gamma}_m$ be the $m$th row of $\m{\Gamma}$. Using \eqref{eq:Cpqjk} we can write the $\sbra{m,n}$th entry of $\m{D}\sbra{j,k}$ in \eqref{eq:Djk} as
\equ{\begin{split} D_{mn}\sbra{j,k}
&= \m{\Gamma}_m\widetilde{\m{C}}\sbra{j,k} \m{\Gamma}_n^T\\
&= \sum_{p\in\m{S}_m}\sum_{q\in\m{S}_n} C_{pq} e^{i2\pi\sbra{jf_{1p} - kf_{1q}}}. \end{split} \label{eq:Dmnjk}}
Note that \eqref{eq:Dmnjk} holds whenever $0\leq j, k \leq n_1-1$ and that $D_{mn}\sbra{j,k}=0$ whenever $m\neq n$. Therefore, when $m\neq n$ we have the following identity:
\equ{\m{A}_{n_1}\sbra{m} \m{C}_{\m{S}_m \m{S}_n} \m{A}_{n_1}^H\sbra{n} = \m{0},}
where $\m{A}_{n_1}\sbra{m}=\mbra{\m{a}_{n_1}\sbra{f_{1p}}}_{p\in\m{S}_m} \in\bC^{n_1\times \abs{\m{S}_m}}$, $m=1,\dots,r'_1$ all have full column rank since $f_{1p}$, $p\in\m{S}_m$ are distinct, and $\m{C}_{\m{S}_m \m{S}_n}$ is a submatrix of $\m{C}$ with rows indexed by $\m{S}_m$ and columns indexed by $\m{S}_n$. It immediately follows that
\equ{\m{C}_{\m{S}_m \m{S}_n} = \m{0} \quad \text{ when } m\neq n, \label{eq:COmOn}}
which means that $\m{C}$ is a block diagonal matrix after properly sorting its rows and columns with respect to $\m{S}_m$, $m=1,\dots,r'$.

Next, we show that $\m{C}$ is also a block diagonal matrix when its rows and columns are sorted in another way. Let $\m{\cP}$ be the permutation matrix satisfying that
\equ{\m{\cP}\m{a}\sbra{\m{f}'} = \m{a}_{n_2}\sbra{f'_{2}} \otimes \sbra{\bigotimes_{l\neq 2} \m{a}_{n_l}\sbra{f'_{l}}} }
for any $\m{f}'\in\bT^d$. Then,
\equ{\m{T}^{\sbra{2}}=\m{\cP}\m{T}\m{\cP}^T }
remains a $d$LT matrix by exchanging the roles of the first and the second dimension of the $d$D frequencies. By viewing $\m{T}^{(2)}$ as an $n_2\times n_2$ block Toeplitz matrix, we can repeat the analysis above. In particular, we can similarly define $\m{\Gamma}'$ and $\m{S}'_m$ according to the partition of $\lbra{\m{f}_{-2,j}}_{j=1}^r$, where $\m{f}_{-2,j}$ denotes $\m{f}_{:j}$ with the second element removed. Then $\m{C}$ is a block diagonal matrix when its rows and columns are sorted based on $\m{S}'_m$.

Now we are ready to show that $\m{C}$ is a diagonal matrix using contradiction. To do so, suppose that $C_{pq}\neq 0$ for some $p\neq q$. By \eqref{eq:COmOn} there must exist some $m$ such that $p,q\in\m{S}_m$. It follows from \eqref{eq:Omegam} that $\m{f}_{-1,p} = \m{f}_{-1,q}=\widetilde{\m{f}}_{-1,m}$. Similarly, there must exist some $m'$ such that $p,q\in\m{S}'_{m'}$ and thus $\m{f}_{-2,p} = \m{f}_{-2,q}$. It therefore holds that $\m{f}_{:p} = \m{f}_{:q}$, which contradicts the assumption that $\m{f}_{:j}$, $j=1,\dots,r$ are distinct.

\subsection{Proof of Proposition \ref{prop:Mepsilony}} \label{append:Mepsilony}

We complete the proof in four steps. In \emph{Step 1}, we show that the optimizer $\sbra{t_{\epsilon}^*,\, \m{T}_{\epsilon}'^*}$ of the problem in \eqref{eq:newsparsemetric} is bounded for $\epsilon\in\sbra{0,1}$. Let $\m{T}_{\epsilon}'^* = \sum_{j=1}^{N'} \lambda_{\epsilon,j}\m{q}_{\epsilon,j}\m{q}_{\epsilon,j}^H$ be the eigen-decomposition of $\m{T}_{\epsilon}'^*$, where the eigenvalues $\lambda_{\epsilon,j}$, $j=1,\dots,N'$ are sorted descendingly. Let also $r_{\epsilon} = \rank\sbra{\m{T}_{\epsilon}'^*}$ and $\overline{p}_{\epsilon,j} = \abs{\m{q}_{\epsilon,j}^H\m{y}'}^2$. Then we have that
{\lentwo{\equa{ t_{\epsilon}^*
&=& \sum_{j=1}^{r_{\epsilon}}\frac{\overline{p}_{\epsilon,j}}{\lambda_{\epsilon,j}}, \label{eq:tepsilon}\\ \cM^{\epsilon}\sbra{\m{y}}
&=& \sum_{j=1}^{N'}\ln\sbra{\lambda_{\epsilon,j} + \epsilon} + \sum_{j=1}^{r_{\epsilon}}\frac{\overline{p}_{\epsilon,j}}{\lambda_{\epsilon,j}}. \label{eq:Mepsilon}
}}}By the optimality of $\lambda_{\epsilon,j}$, it holds that $\frac{\partial \cM^{\epsilon}\sbra{\m{y}}}{\partial \lambda_{\epsilon,j}} = \frac{1}{\lambda_{\epsilon,j} + \epsilon} - \frac{\overline{p}_{\epsilon,j}}{\lambda_{\epsilon,j}^2}=0$ and thus
\equ{\overline{p}_{\epsilon,j} = \frac{\lambda_{\epsilon,j}^2}{\lambda_{\epsilon,j} + \epsilon}\in\sbra{\lambda_{\epsilon,j} - \epsilon,\, \lambda_{\epsilon,j}}, \quad j = 1,\dots,r_{\epsilon}. \label{eq:pbarleqlambda}}
Inserting \eqref{eq:pbarleqlambda} into \eqref{eq:tepsilon}, we have that $t_{\epsilon}^* \leq r_{\epsilon}\leq N'$ is bounded. Moreover, by \eqref{eq:pbarleqlambda} we also have that
\equ{\tr\sbra{\m{T}_{\epsilon}'^*} = \sum_{j=1}^{r_{\epsilon}} \lambda_{\epsilon,j} \leq \sum_{j=1}^{r_{\epsilon}} \overline{p}_{\epsilon,j} + r_{\epsilon}\epsilon \leq \twon{\m{y}'}^2 + N'}
and hence, $\m{T}_{\epsilon}'^*$ is bounded as well.

In \emph{Step 2}, we show that $\lambda_{\epsilon,\cM\sbra{\m{y}'}}\geq c$ for any $\epsilon\in\sbra{0,1}$, where $c>0$ is a constant. To do so, suppose that there exist $\epsilon_j\in\sbra{0,1}$, $j=1,2,\dots$ such that $\lambda_{\epsilon_j,\cM\sbra{\m{y}'}}< \frac{1}{j}$. Since $\sbra{t_{\epsilon}^*,\, \m{T}_{\epsilon}'^*}$ is bounded, without loss of generality, we assume that $\sbra{t_{\epsilon_j}^*,\, \m{T}_{\epsilon_j}'^*}\rightarrow \sbra{t_0^*,\, \m{T}_{0}'^*}$, as $j\rightarrow+\infty$. As a result, $\rank\sbra{\m{T}_{0}'^*}< \cM\sbra{\m{y}'}$ since $\lambda_{\epsilon_j,\cM\sbra{\m{y}'}}\rightarrow0$, as $j\rightarrow+\infty$. On the other hand, it must hold that
\equ{ \begin{bmatrix}t_0^* & \m{y}'^H \\ \m{y}' & \m{T}_{0}'^* \end{bmatrix} =  \lim_{j\rightarrow+\infty}\begin{bmatrix}t_{\epsilon_j}^* & \m{y}'^H \\ \m{y}' & \m{T}_{\epsilon_j}'^* \end{bmatrix} \geq \m{0}. }
Therefore, $\sbra{t_0^*, \; \m{T}_{0}'^*}$ is a feasible solution to the problem in \eqref{eq:rankmin_problem2}. It follows that $\cM\sbra{\m{y}'} \leq \rank\sbra{\m{T}_{0}'^*}$, contradicting the fact that $\rank\sbra{\m{T}_{0}'^*}< \cM\sbra{\m{y}'}$ as shown previously.

In \emph{Step 3}, we prove the first part of the proposition. By \eqref{eq:Mepsilon} and the inequality $\lambda_{\epsilon,\cM\sbra{\m{y}'}}\geq c$ shown in \emph{Step 2} we have that
\equ{\begin{split}\cM^{\epsilon}\sbra{\m{y}}
&\geq \sum_{j=1}^{N'}\ln\sbra{\lambda_{\epsilon,j} + \epsilon} \\
&\geq \sbra{N'-\cM\sbra{\m{y}'}}\ln\epsilon + c_1,
\end{split} \label{eq:Mlowbound}}
where $c_1=\sum_{j=1}^{\cM\sbra{\m{y}'}}\ln c$ is a constant independent of $\epsilon$.
On the other hand, since any optimizer of the problem in \eqref{eq:rankmin_problem2} is a feasible solution to the problem in \eqref{eq:newsparsemetric}, it is easy to see that
\equ{\begin{split}\cM^{\epsilon}\sbra{\m{y}}
&\leq \sbra{N'-\cM\sbra{\m{y}'}}\ln\epsilon +c_2
\end{split} \label{eq:Muppbound}}
where $c_2$ is also a constant. Combining \eqref{eq:Mlowbound} and \eqref{eq:Muppbound} proves the first part.

To prove the second part of the proposition, in \emph{Step 4}, we refine \eqref{eq:Mlowbound} as
\equ{\begin{split}\cM^{\epsilon}\sbra{\m{y}}
\geq &\sbra{N'-\cM\sbra{\m{y}'}}\ln\epsilon \\
& + \sum_{j=\cM\sbra{\m{y}'}+1}^{N'} \ln \sbra{\frac{\lambda_{\epsilon,j}}{\epsilon} + 1} + c_1.
\end{split} \label{eq:Mlowbound2}}
Combining \eqref{eq:Mlowbound2} and \eqref{eq:Muppbound}, we have that
\equ{\begin{split}
&\ln \sbra{\frac{\lambda_{\epsilon,\cM\sbra{\m{y}'}+1}}{\epsilon}+1} \\
& \leq \sum_{j=\cM\sbra{\m{y}'}+1}^{N'} \ln \sbra{\frac{\lambda_{\epsilon,j}}{\epsilon}+1} \\
& \leq c_2-c_1. \end{split}}
It follows that
\equ{\lambda_{\epsilon,N'}\leq \dots \leq \lambda_{\epsilon,\cM\sbra{\m{y}'}+1} \leq \sbra{e^{c_2-c_1}-1}\epsilon. \label{eq:lambdaepsilon}}
Moreover, by \eqref{eq:lambdaepsilon} any cluster point of $\m{T}_{\epsilon}'^*$ at $\epsilon=0$ has rank no greater than $\cM\sbra{\m{y}'}$. On the other hand, the rank is no less than $\cM\sbra{\m{y}'}$ by the result in \emph{Step 2}. This observation completes the proof.

\bibliographystyle{IEEEtran}

\begin{thebibliography}{10}
\providecommand{\url}[1]{#1}
\csname url@samestyle\endcsname
\providecommand{\newblock}{\relax}
\providecommand{\bibinfo}[2]{#2}
\providecommand{\BIBentrySTDinterwordspacing}{\spaceskip=0pt\relax}
\providecommand{\BIBentryALTinterwordstretchfactor}{4}
\providecommand{\BIBentryALTinterwordspacing}{\spaceskip=\fontdimen2\font plus
\BIBentryALTinterwordstretchfactor\fontdimen3\font minus
  \fontdimen4\font\relax}
\providecommand{\BIBforeignlanguage}[2]{{%
\expandafter\ifx\csname l@#1\endcsname\relax
\typeout{** WARNING: IEEEtran.bst: No hyphenation pattern has been}%
\typeout{** loaded for the language `#1'. Using the pattern for}%
\typeout{** the default language instead.}%
\else
\language=\csname l@#1\endcsname
\fi
#2}}
\providecommand{\BIBdecl}{\relax}
\BIBdecl

\bibitem{yang2015generalized}
Z.~Yang, L.~Xie, and P.~Stoica, ``Generalized {Vandermonde} decomposition and
  its use for multi-dimensional super-resolution,'' in \emph{IEEE International
  Symposium on Information Theory (ISIT)}, 2015, pp. 2011--2015.

\bibitem{caratheodory1911zusammenhang}
C.~Carath\'{e}odory and L.~Fej{\'e}r, ``{{\"U}ber den Zusammenhang der Extremen
  von harmonischen Funktionen mit ihren Koeffizienten und {\"u}ber den
  Picard-Landau'schen Satz},'' \emph{Rendiconti del Circolo Matematico di
  Palermo (1884-1940)}, vol.~32, no.~1, pp. 218--239, 1911.

\bibitem{pisarenko1973retrieval}
V.~F. Pisarenko, ``The retrieval of harmonics from a covariance function,''
  \emph{Geophysical Journal International}, vol.~33, no.~3, pp. 347--366, 1973.

\bibitem{stoica2005spectral}
P.~Stoica and R.~L. Moses, \emph{Spectral analysis of signals}.\hskip 1em plus
  0.5em minus 0.4em\relax Pearson/Prentice Hall Upper Saddle River, NJ, 2005.

\bibitem{hua1993pencil}
Y.~Hua, ``A pencil-{MUSIC} algorithm for finding two-dimensional angles and
  polarizations using crossed dipoles,'' \emph{IEEE Transactions on Antennas
  and Propagation}, vol.~41, no.~3, pp. 370--376, 1993.

\bibitem{haardt19952d}
M.~Haardt, M.~D. Zoltowski, C.~P. Mathews, and J.~A. Nossek, ``{2D unitary
  ESPRIT for efficient 2D parameter estimation},'' in \emph{IEEE International
  Conference on Acoustics, Speech, and Signal Processing (ICASSP)}, vol.~3,
  1995, pp. 2096--2099.

\bibitem{li1992two}
J.~Li and R.~Compton~Jr, ``Two-dimensional angle and polarization estimation
  using the {ESPRIT} algorithm,'' \emph{IEEE Transactions on Antennas and
  Propagation}, vol.~40, no.~5, pp. 550--555, 1992.

\bibitem{hua1992estimating}
Y.~Hua, ``Estimating two-dimensional frequencies by matrix enhancement and
  matrix pencil,'' \emph{IEEE Transactions on Signal Processing}, vol.~40,
  no.~9, pp. 2267--2280, 1992.

\bibitem{liu2002almost}
X.~Liu and N.~D. Sidiropoulos, ``Almost sure identifiability of constant
  modulus multidimensional harmonic retrieval,'' \emph{IEEE Transactions on
  Signal Processing}, vol.~50, no.~9, pp. 2366--2368, 2002.

\bibitem{liu2006eigenvector}
J.~Liu and X.~Liu, ``An eigenvector-based approach for multidimensional
  frequency estimation with improved identifiability,'' \emph{IEEE Transactions
  on Signal Processing}, vol.~54, no.~12, pp. 4543--4556, 2006.

\bibitem{liu2007multidimensional}
J.~Liu, X.~Liu, and X.~Ma, ``Multidimensional frequency estimation with finite
  snapshots in the presence of identical frequencies,'' \emph{IEEE Transactions
  on Signal Processing}, vol.~55, no.~11, pp. 5179--5194, 2007.

\bibitem{candes2013towards}
E.~J. Cand{\`e}s and C.~Fernandez-Granda, ``Towards a mathematical theory of
  super-resolution,'' \emph{Communications on Pure and Applied Mathematics},
  vol.~67, no.~6, pp. 906--956, 2014.

\bibitem{candes2006robust}
E.~Cand{\`e}s, J.~Romberg, and T.~Tao, ``{Robust uncertainty principles: Exact
  signal reconstruction from highly incomplete frequency information},''
  \emph{IEEE Transactions on Information Theory}, vol.~52, no.~2, pp. 489--509,
  2006.

\bibitem{tang2012compressed}
G.~Tang, B.~N. Bhaskar, P.~Shah, and B.~Recht, ``Compressed sensing off the
  grid,'' \emph{IEEE Transactions on Information Theory}, vol.~59, no.~11, pp.
  7465--7490, 2013.

\bibitem{yang2014continuous}
Z.~Yang and L.~Xie, ``Continuous compressed sensing with a single or multiple
  measurement vectors,'' in \emph{IEEE Workshop on Statistical Signal
  Processing (SSP)}, 2014, pp. 308--311.

\bibitem{aleksanyan1944real}
S.~Aleksanyan, A.~Apozyan, V.~Z. Dumanyan, K.~A. Khachatryan, E.~Nazari,
  A.~Pahlevanyan, and H.~Rostami, ``Real and complex analysis,''
  \emph{Mathematics in Armenia}, vol.~54, p.~21, 1944.

\bibitem{chandrasekaran2012convex}
V.~Chandrasekaran, B.~Recht, P.~A. Parrilo, and A.~S. Willsky, ``The convex
  geometry of linear inverse problems,'' \emph{Foundations of Computational
  Mathematics}, vol.~12, no.~6, pp. 805--849, 2012.

\bibitem{candes2013super}
E.~J. Cand{\`e}s and C.~Fernandez-Granda, ``Super-resolution from noisy data,''
  \emph{Journal of Fourier Analysis and Applications}, vol.~19, no.~6, pp.
  1229--1254, 2013.

\bibitem{bhaskar2013atomic}
B.~N. Bhaskar, G.~Tang, and B.~Recht, ``Atomic norm denoising with applications
  to line spectral estimation,'' \emph{IEEE Transactions on Signal Processing},
  vol.~61, no.~23, pp. 5987--5999, 2013.

\bibitem{yang2015gridless}
Z.~Yang and L.~Xie, ``On gridless sparse methods for line spectral estimation
  from complete and incomplete data,'' \emph{IEEE Transactions on Signal
  Processing}, vol.~63, no.~12, pp. 3139--3153, 2015.

\bibitem{tang2015near}
G.~Tang, B.~N. Bhaskar, and B.~Recht, ``Near minimax line spectral
  estimation,'' \emph{IEEE Transactions on Information Theory}, vol.~61, no.~1,
  pp. 499--512, 2015.

\bibitem{yang2014exact}
\BIBentryALTinterwordspacing
Z.~Yang and L.~Xie, ``Exact joint sparse frequency recovery via optimization
  methods,'' 2014. [Online]. Available: \url{http://arxiv.org/abs/1405.6585}
\BIBentrySTDinterwordspacing

\bibitem{yang2014enhancing}
Z.~Yang and L.~Xie, ``Enhancing sparsity and resolution via reweighted atomic norm
  minimization,'' \emph{IEEE Transactions on Signal Processing}, vol.~64,
  no.~4, pp. 995--1006, 2016.

\bibitem{xu2014precise}
W.~Xu, J.-F. Cai, K.~V. Mishra, M.~Cho, and A.~Kruger, ``Precise semidefinite
  programming formulation of atomic norm minimization for recovering
  $d$-dimensional ($d\geq2$) off-the-grid frequencies,'' in \emph{Information
  Theory and Applications Workshop (ITA)}, 2014, pp. 1--4.

\bibitem{dumitrescu2007positive}
B.~Dumitrescu, \emph{Positive trigonometric polynomials and signal processing
  applications}.\hskip 1em plus 0.5em minus 0.4em\relax Springer, 2007.

\bibitem{bendory2015super}
T.~Bendory, S.~Dekel, and A.~Feuer, ``Super-resolution on the sphere using
  convex optimization,'' \emph{IEEE Transactions on Signal Processing},
  vol.~63, no.~9, pp. 2253--2262, 2015.

\bibitem{heckel2014super}
R.~Heckel, V.~I. Morgenshtern, and M.~Soltanolkotabi, ``Super-resolution
  radar,'' \emph{arXiv preprint arXiv:1411.6272}, 2014.

\bibitem{sidiropoulos2001generalizing}
N.~D. Sidiropoulos, ``Generalizing {Caratheodory's} uniqueness of harmonic
  parameterization to {$N$} dimensions,'' \emph{IEEE Transactions on
  Information Theory}, vol.~47, no.~4, pp. 1687--1690, 2001.

\bibitem{georgiou2000signal}
T.~T. Georgiou, ``Signal estimation via selective harmonic amplification:
  {MUSIC, Redux},'' \emph{IEEE Transactions on Signal Processing}, vol.~48,
  no.~3, pp. 780--790, 2000.

\bibitem{georgiou2007caratheodory}
T.~T. Georgiou, ``{The Carath{\'e}odory--Fej{\'e}r--Pisarenko decomposition and its
  multivariable counterpart},'' \emph{IEEE Transactions on Automatic Control},
  vol.~52, no.~2, pp. 212--228, 2007.

\bibitem{gurvits2002largest}
L.~Gurvits and H.~Barnum, ``Largest separable balls around the maximally mixed
  bipartite quantum state,'' \emph{Physical Review A}, vol.~66, no.~6, p.
  062311, 2002.

\bibitem{chi2015compressive}
Y.~Chi and Y.~Chen, ``Compressive two-dimensional harmonic retrieval via atomic
  norm minimization,'' \emph{IEEE Transactions on Signal Processing}, vol.~63,
  no.~4, pp. 1030--1042, 2015.

\bibitem{kronecker1895leopold}
L.~Kronecker, \emph{Leopold Kronecker's werke}.\hskip 1em plus 0.5em minus
  0.4em\relax BG Teubner, 1895.

\bibitem{rochberg1987toeplitz}
R.~Rochberg, ``{Toeplitz and Hankel operators on the Paley-Wiener space},''
  \emph{Integral Equations and Operator Theory}, vol.~10, no.~2, pp. 187--235,
  1987.

\bibitem{andersson2015general}
F.~Andersson and M.~Carlsson, ``On general domain truncated correlation and
  convolution operators with finite rank,'' \emph{Integral Equations and
  Operator Theory}, pp. 1--32, 2015.

\bibitem{chen2014robust}
Y.~Chen and Y.~Chi, ``Robust spectral compressed sensing via structured matrix
  completion,'' \emph{IEEE Transactions on Information Theory}, vol.~60,
  no.~10, pp. 6576--6601, 2014.

\bibitem{andersson2014new}
F.~Andersson, M.~Carlsson, J.-Y. Tourneret, and H.~Wendt, ``A new frequency
  estimation method for equally and unequally spaced data,'' \emph{IEEE
  Transactions on Signal Processing}, vol.~62, no.~21, pp. 5761--5774, 2014.

\bibitem{horn2012matrix}
R.~A. Horn and C.~R. Johnson, \emph{Matrix analysis}.\hskip 1em plus 0.5em
  minus 0.4em\relax Cambridge University Press, 2012.

\bibitem{ellis1992factorization}
R.~L. Ellis and D.~C. Lay, ``{Factorization of finite rank Hankel and Toeplitz
  matrices},'' \emph{Linear Algebra and Its Applications}, vol. 173, pp.
  19--38, 1992.

\bibitem{jiang2001almost}
T.~Jiang, N.~D. Sidiropoulos, and J.~M. ten Berge, ``Almost-sure
  identifiability of multidimensional harmonic retrieval,'' \emph{IEEE
  Transactions on Signal Processing}, vol.~49, no.~9, pp. 1849--1859, 2001.

\bibitem{recht2007guaranteed}
B.~Recht, M.~Fazel, and P.~Parrilo, ``{Guaranteed minimum-rank solutions of
  linear matrix equations via nuclear norm minimization},'' \emph{SIAM Review},
  vol.~52, no.~3, pp. 471--501, 2007.

\bibitem{toh1999sdpt3}
K.-C. Toh, M.~J. Todd, and R.~H. T{\"u}t{\"u}nc{\"u}, ``{SDPT3--a MATLAB
  software package for semidefinite programming, version 1.3},''
  \emph{Optimization Methods and Software}, vol.~11, no. 1-4, pp. 545--581,
  1999.

\bibitem{boyd2011distributed}
S.~Boyd, N.~Parikh, E.~Chu, B.~Peleato, and J.~Eckstein, ``Distributed
  optimization and statistical learning via the alternating direction method of
  multipliers,'' \emph{Foundations and Trends{\textregistered} in Machine
  Learning}, vol.~3, no.~1, pp. 1--122, 2011.

\end{thebibliography}


\end{document}